\documentclass[conference]{IEEEtran}

\usepackage{amsmath,amssymb,amsthm}
\usepackage{mathtools}
\usepackage{enumerate}
\usepackage{graphicx}
\usepackage{booktabs}
\usepackage{array}
\usepackage{hyperref}
\usepackage[table]{xcolor}
\usepackage{microtype}
\usepackage{balance}
\usepackage{tikz}
\usepackage[misc]{ifsym}
\usetikzlibrary{arrows.meta, positioning, shapes.geometric, fit, backgrounds}
\usepackage[ruled,linesnumbered]{algorithm2e}

\SetCommentSty{AlgoCommentSty}
\definecolor{TableHeader}{RGB}{232,238,247}
\definecolor{TableStripe}{RGB}{247,249,252}
\definecolor{TableHighlight}{RGB}{232,246,238}
\definecolor{TableRule}{RGB}{120,130,145}

\newcommand{\tstriped}{\rowcolor{TableStripe}}
\newcommand{\thighlight}{\rowcolor{TableHighlight}}

\newtheorem{theorem}{Theorem}
\newtheorem{proposition}[theorem]{Proposition}

\newtheorem{corollary}[theorem]{Corollary}
\newtheorem{definition}{Definition}
\newtheorem{assumption}{Assumption}

\newcommand{\ST}{\mathcal{S}_T}

\IEEEoverridecommandlockouts

\begin{document}

\title{Measuring the Unmeasurable: Markov Chain Reliability for LLM Agents}

\author{\IEEEauthorblockN{Phat T. Tran-Truong \href{https://orcid.org/0000-0003-3199-6333}{\includegraphics[scale=0.004]{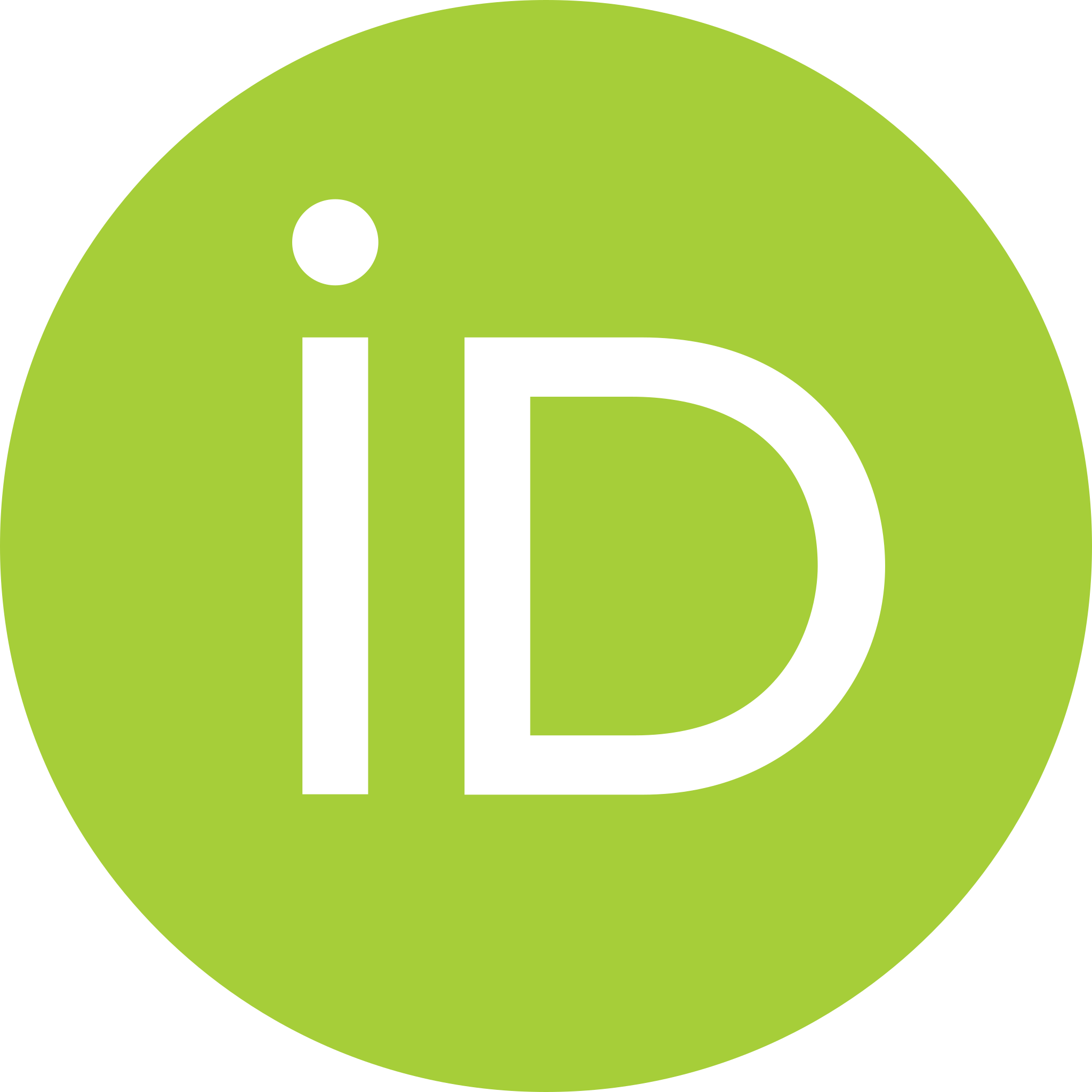}}, Xuan-Bach Le$^{\text{(\Letter)}}$ \href{https://orcid.org/0009-0003-6848-5403}{\includegraphics[scale=0.004]{sugg.png}}}
\IEEEauthorblockA{\textit{Faculty of Computer Science and Engineering} \\
\textit{Ho Chi Minh City University of Technology (HCMUT), VNU-HCM}\\
Ho Chi Minh City, Vietnam \\
\{phatttt, lexuanbach\}@hcmut.edu.vn}
}

\maketitle
\begingroup\renewcommand\thefootnote{}
\footnotetext{
$^{\text{(\Letter)}}$ Corresponding Author: Xuan-Bach Le (lexuanbach@hcmut.edu.vn)
}
\endgroup

\begin{abstract}
Large language model (LLM) agents increasingly operate as sequential
software systems, but their reliability is often summarized by scalar
benchmark metrics. Metrics such as pass$@k$, pass$^k$, and the
reliability decay curve (RDC) are useful summaries, but they do not
identify the success-time
distribution being estimated, test whether traces support that
distribution, or quantify finite-trace uncertainty. We present
\textsc{TraceToChain}, a reproducible pipeline that fits agent
execution traces to an absorbing discrete-time Markov chain (DTMC),
$\hat M=(\hat Q,\hat R_\oplus,\hat R_\ominus)$, with explicit
diagnostics and uncertainty. The pipeline builds an
automatic cluster taxonomy, estimates transitions with Laplace-smoothed
maximum-likelihood estimation (MLE), checks fit with a composite Akaike
information criterion (AIC) and Kolmogorov--Smirnov (KS)
goodness-of-fit certificate, and reports Dirichlet-posterior credible
intervals and non-parametric bootstrap intervals. We adapt classical
reliability mathematics (Kemeny--Snell~\cite{kemenysnell},
Cheung~\cite{cheung1980}, Goel--Okumoto~\cite{goelokt}) to agent
traces. The resulting first-passage view reconciles metrics usually
reported separately: pass$@k$, pass$^k$,
and the RDC are projections of one success-time distribution. On seven
controlled MAST-style frameworks with a strict 50/50 fit/test protocol,
held-out empirical RDCs overlay their analytic counterparts with
max $L_\infty^{\mathrm{RDC}} = 0.053$ (median $0.048$). A two-sample KS
test on the first-passage cumulative distribution function (CDF)
accepts the fitted chain with $p>0.05$ on $7/7$ frameworks
(min $p = 0.78$), and per-entry $95\%$ posterior and bootstrap
intervals agree to $\approx\!0.01$ at the median. 

% The anonymous
% artifact is available at \url{https://anonymous.4open.science/r/TraceToChain-EEE6}.
\end{abstract}

\begin{IEEEkeywords}
software reliability, large language models, agent systems,
absorbing Markov chains, fundamental matrix, pass@k, NHPP, reliability decay curve
\end{IEEEkeywords}

%% ================================================================

\section{Introduction}
\label{sec:intro}

Large language model (LLM) agents are deployed as sequential software
systems, yet they are often evaluated by scalar outcomes. A deployed
agent plans, calls tools, observes results, retries, and terminates in
success or failure, as in ReAct-style and tool-using
designs~\cite{yao2023react,schick2023toolformer}. This shift toward
autonomous interactive systems is increasingly visible across agent
surveys~\cite{wang2024agentsurvey}. Reliability for such systems is a
question about first-passage behavior, finite-trace uncertainty, and
model fit, not only a terminal pass rate. These are standard
dependability and site reliability
engineering (SRE)
concerns~\cite{avizienis2004dependability,beyer2016sre}, but current
LLM-agent benchmark summaries rarely state the reliability
distribution being estimated or audit whether it fits the traces.

This gap appears immediately in deployment. A fintech team may ship a
ReAct-style refund agent~\cite{yao2023react} after measuring
$\mathrm{pass}@1 = 0.72$ on $\tau$-bench~\cite{taub2024}, then need to
estimate reliability under a larger step budget, the end-to-end effect
of a fallback search tool, or the relation between an internal
$\mathrm{pass}^5$ dashboard and SRE mean-time-between-failures (MTBF)
language. These are ordinary questions for a sequential system, but a
single pass rate does not answer them.

We ask whether empirical agent traces can be turned into an audited
reliability model for such queries without rerunning a benchmark for
each horizon, local edit, or metric convention.

\vspace{0.35em}
\noindent\begin{tikzpicture}
\node[
  draw=black!35,
  fill=gray!4,
  rounded corners=3pt,
  line width=0.45pt,
  inner sep=5pt,
  text width=0.94\columnwidth,
  align=left
] {
\begin{minipage}{0.94\columnwidth}
\raggedright
\textbf{Reliability question.}
Can empirical LLM-agent traces become an audited model
for horizon queries, uncertainty, goodness-of-fit, and metric
reconciliation?
\end{minipage}
};
\end{tikzpicture}

Figure~\ref{fig:motivating} summarizes the gap. Scalar metrics leave
horizon queries, local counterfactuals, and denominator choices outside
the reported number. A fitted chain provides the shared object behind
those questions by treating pass$@k$, pass$^k$, and the reliability
decay curve (RDC) as projections of one
success-time distribution, where $k$ denotes repeated trials.

\begin{figure}[!t]
\centering
\begin{tikzpicture}[
    >=Stealth,
    font=\scriptsize\sffamily,
    hdr/.style={font=\bfseries\sffamily\scriptsize, align=center},
    q/.style={rectangle, rounded corners, draw=orange!80!black,
              fill=orange!10, line width=0.7pt, align=left,
              text width=2.15cm, inner sep=2.5pt, minimum height=1.25cm},
    s/.style={rectangle, rounded corners, draw=red!60!black,
              fill=red!6, line width=0.7pt, align=left,
              text width=2.15cm, inner sep=2.5pt, minimum height=1.25cm},
    a/.style={rectangle, rounded corners, draw=green!55!black,
              fill=green!6, line width=0.7pt, align=left,
              text width=2.15cm, inner sep=2.5pt, minimum height=1.25cm},
    lbl/.style={font=\scriptsize\itshape, text=gray!50!black},
    arr/.style={->, draw=gray!70, line width=0.7pt},
    xscale=1.0, yscale=1.0,
]
    % column centers
    \def\colQ{0}
    \def\colS{2.75}
    \def\colA{5.5}
    % row y-coordinates
    \def\rowH{0.85}
    \def\rowA{-0.5}
    \def\rowB{-2.05}
    \def\rowC{-3.6}

    % Headers
    \node[hdr] at (\colQ, \rowH)
        {Question};
    \node[hdr] at (\colS, \rowH)
        {Scalar metric};
    \node[hdr] at (\colA, \rowH)
        {Fitted chain $\hat M$};

    % Row 1: budget
    \node[q] (q1) at (\colQ, \rowA)
        {\textbf{Q1.} If $d{=}8\!\to\!16$, what reliability?};
    \node[s] (s1) at (\colS, \rowA)
        {Re-run at each $d$. No interval or extrapolation.};
    \node[a] (a1) at (\colA, \rowA)
        {$\mathcal R(d)$ on any grid
         (Prop.~\ref{thm:closed}) + 95\% interval.};

    % Row 2: fallback tool
    \node[q] (q2) at (\colQ, \rowB)
        {\textbf{Q2.} Add fallback search: $\Delta R_\infty$?};
    \node[s] (s2) at (\colS, \rowB)
        {Re-run full benchmark for one new point.};
    \node[a] (a2) at (\colA, \rowB)
        {Bound $\Delta R_\infty$
         locally (Prop.~\ref{thm:perturb}).};

    % Row 3: metric zoo
    \node[q] (q3) at (\colQ, \rowC)
        {\textbf{Q3.} pass$^5$ vs.\ SRE-MTBF: which is right?};
    \node[s] (s3) at (\colS, \rowC)
        {Two denominators and ad hoc reconciliation.};
    \node[a] (a3) at (\colA, \rowC)
        {Both are marginals of one success-time distribution
         (Prop.~\ref{thm:unify}).};

    \draw[arr] (q1) -- (s1);
    \draw[arr] (s1) -- (a1);
    \draw[arr] (q2) -- (s2);
    \draw[arr] (s2) -- (a2);
    \draw[arr] (q3) -- (s3);
    \draw[arr] (s3) -- (a3);
\end{tikzpicture}
\caption{Motivating example: a fitted absorbing chain turns one trace
         corpus into horizon, perturbation, and metric-reconciliation
         queries.}
\label{fig:motivating}
\end{figure}
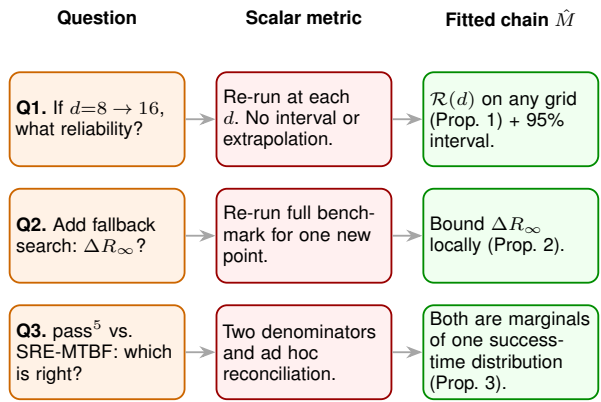

We make this distribution explicit by modeling each execution as a
trajectory in an absorbing discrete-time Markov chain (DTMC), with
transient states for intermediate behavior and absorbing states for
terminal success or failure. The fitted chain
$\hat M=(\hat Q,\hat R_\oplus,\hat R_\ominus)$ records intermediate
transitions in $\hat Q$ and state-specific exits in $\hat R_\oplus$
and $\hat R_\ominus$. Reliability at a step horizon, local
perturbation analysis, and metric reconciliation then become
first-passage queries on the same model
(Props.~\ref{thm:closed}, \ref{thm:perturb}, and \ref{thm:unify}).

The main challenge is auditing: the chain must be fit from empirical
traces, tested for goodness-of-fit, assigned finite-sample uncertainty,
and connected to established metrics. We therefore adapt classical
reliability mathematics to agent traces rather than adding another
benchmark-specific score. Absorbing chains and the fundamental matrix
$N=(I-Q)^{-1}$ are due to Kemeny and Snell~\cite{kemenysnell} and were
used in software reliability by Cheung~\cite{cheung1980}. Perturbation
analysis~\cite{stewart1998,hornjohnson2013} and quantitative
reliability modeling~\cite{trivedi2017reliability} provide the broader
engineering context. The Goel--Okumoto and
Musa--Okumoto models connect to reliability growth
models~\cite{goelokt,musa1987}, while stochastic-monotone chains
explain RDC shape~\cite{karlin1968tp,keilsonKester1977}. In the
rare-failure regime, the fitted agent chain recovers the Goel--Okumoto
limit (Proposition~\ref{thm:nhpp}).

We instantiate this view as \textsc{TraceToChain}, an audited pipeline
that fits an absorbing DTMC, attaches diagnostics and uncertainty, and
reports common LLM-agent metrics as views of one success first-passage
distribution. Our contributions are:
\begin{enumerate}[(1)]
  \item \emph{Trace-to-chain modeling.}
  \textsc{TraceToChain} (\S\ref{sec:state_construction},
  Algorithms~\ref{alg:trace}--\ref{alg:gof}) maps traces to a fitted
  absorbing DTMC with a composite Akaike information criterion (AIC)
  and Kolmogorov--Smirnov (KS) goodness-of-fit certificate.
  \item \emph{Uncertainty and diagnostics.}
  \S\ref{sec:uq} gives Dirichlet-posterior and trace-level bootstrap
  intervals for the fitted chain and derived reliability quantities.
  \item \emph{Held-out validation.}
  A strict held-out validation on seven controlled MAST-style frameworks
  (simulation study SS9, \S\ref{sec:ss9}) obtains maximum RDC error
  $L_\infty^{\mathrm{RDC}}=0.053$ (median $0.048$) and held-out KS
  $p>0.05$ on $7/7$ frameworks (min $p=0.78$).
  \item \emph{Metric unification.}
  pass$@k$, pass$^k$, and the RDC become projections of one success
  first-passage distribution, making their denominators explicit.
\end{enumerate}
Raw SWE-bench and $\tau$-bench trajectories require step-level feature
data and remain future validation targets.

The paper proceeds as follows. Related Work (\S\ref{sec:related})
positions the problem, and Formalism (\S\ref{sec:formal}) and State
Construction (\S\ref{sec:state_construction}) define and fit the
absorbing-chain model. Analytic Tools (\S\ref{sec:core}) and
Extensions (\S\ref{sec:ext}) derive the main reliability results.
Simulation Validation (\S\ref{sec:sim}) and Empirical Case Studies
(\S\ref{sec:case_studies}) evaluate the machinery, and the final
sections discuss limits and implications
(\S\ref{sec:threats}--\S\ref{sec:conc}).
\section{Related Work}
\label{sec:related}

Prior work provides many of the ingredients for deployment reliability
decisions, but not yet their combination. Reliability engineering
contributes the vocabulary, agent benchmarks provide the trace setting,
and probabilistic verification supplies model-based reasoning. The
missing piece is an audited fit from LLM-agent traces: a model with
uncertainty, diagnostics, and shared semantics for common metrics.

\begin{table*}[!t]
\centering
\caption{Comparison of \textsc{TraceToChain} with related evaluation
and verification frameworks.}
\label{tab:method-comparison}
\scriptsize
\setlength{\tabcolsep}{3pt}
\renewcommand{\arraystretch}{1.12}
\begin{tabular}{@{}>{\raggedright\arraybackslash}p{0.22\textwidth}
                >{\raggedright\arraybackslash}p{0.15\textwidth}
                >{\raggedright\arraybackslash}p{0.2\textwidth}
                >{\raggedright\arraybackslash}p{0.22\textwidth}
                >{\raggedright\arraybackslash}p{0.1\textwidth}@{}}
\toprule
\textbf{Approach} & \textbf{Input} & \textbf{Reliability semantics} &
\textbf{Diagnostics / uncertainty} & \textbf{LLM-agent fit} \\
\midrule
\mbox{Classical SRE / NHPP~\cite{cheung1980,goelokt}} &
Aggregate failures &
Growth / failure intensity &
Aggregate intervals, no trace states &
Indirect \\
\tstriped
\mbox{Agent benchmark metrics~\cite{chen2021codex,rdc2025}} &
Outcomes / traces &
Success, pass$@k$, pass$^k$, RDC &
Metric-specific, limited fit checks &
Direct / split \\
\mbox{PRISM / Storm~\cite{prism2011,storm2017}} &
User model &
PCTL queries &
Verification, no trace fitting &
Indirect \\
\tstriped
\mbox{TriCEGAR / ProbGuard~\cite{tricegar2026,pro2guard2025}} &
Empirical traces &
Learned Markov abstraction &
Refinement diagnostics &
Related \\
\thighlight
\mbox{\textbf{\textsc{TraceToChain} (this work)}} &
\textbf{LLM-agent traces} &
\textbf{First-passage DTMC reliability} &
\textbf{AIC/KS diagnostics, posterior/bootstrap intervals} &
\textbf{Direct} \\
\bottomrule
\end{tabular}
\end{table*}

\paragraph{Classical SRE.}
Classical software reliability gives this paper its vocabulary and
analytic tools. The dependability taxonomy of Avi{\v z}ienis et
al. separates reliability, availability, safety, and failure
semantics~\cite{avizienis2004dependability}, while SRE
practice~\cite{beyer2016sre} turns such concepts into measurable
service objectives. State-based reliability modeling is equally
central: Cheung~\cite{cheung1980} used Markov chains for program-level
software reliability, Trivedi and Bobbio~\cite{trivedi2017reliability}
give a broader engineering treatment of reliability and availability
modeling, and Kemeny and Snell~\cite{kemenysnell} developed the
absorbing-chain framework and fundamental matrix $N=(I-Q)^{-1}$ used
by our closed-form reliability quantities. Reliability-growth models
provide a complementary aggregate view through the Goel--Okumoto NHPP
model~\cite{goelokt} and the Musa--Okumoto
formulation~\cite{musa1987}. These tools are not, by themselves, a
trace-level evaluation method for LLM agents; Proposition~\ref{thm:nhpp}
connects them by recovering the NHPP form as a rare-failure scaling
limit of the fitted agent chain.

\paragraph{Agent benchmarks.}
Agent benchmarks show why trace-level reliability semantics are
needed. ReAct~\cite{yao2023react}, Toolformer~\cite{schick2023toolformer},
and Reflexion~\cite{shinn2023reflexion} established agent
trajectories that interleave reasoning, tool use, and revision steps;
ToolBench~\cite{qin2023toolllm} and $\tau$-bench~\cite{taub2024} make
API choice, tool errors, and tool-agent-user interaction part of the
evaluation surface. SWE-bench~\cite{swebench2024},
WebArena~\cite{zhou2023webarena}, ALFWorld~\cite{shridhar2021alfworld},
and Voyager~\cite{wang2023voyager} move evaluation further toward
software, web, embodied, and open-ended environments where success is
long-horizon and context-dependent. MAST~\cite{mast2025} diagnoses
multi-agent failure modes, and AgentBench~\cite{agentbench2024}
broadens the interactive evaluation space beyond one task family.
Together, these benchmarks expose the trace setting; they do not by
themselves unify the reported metrics. pass$@k$~\cite{chen2021codex}
originated as a repeated-sampling success metric for code generation,
and agent evaluation later adapted it alongside the reliability decay
curve (RDC)~\cite{rdc2025}. In our formalism, the RDC is exactly the
$d$-step reliability $R(d)$ from Definition~\ref{def:reliab}.

\paragraph{Formal methods for agents.}
Formal methods supply the model-based reasoning layer. Probabilistic
computation tree logic (PCTL)~\cite{hansson1994pctl} supports
reasoning about probabilistic time and reliability properties, and
model-checking texts~\cite{baier2008principles} place such logics in a
broader verification framework. Given a probabilistic model,
PRISM~\cite{prism2011} and Storm~\cite{storm2017} can verify PCTL
properties. More recent trace-driven frameworks, including
TriCEGAR~\cite{tricegar2026} and ProbGuard~\cite{pro2guard2025}, learn
DTMC or Markov decision process (MDP) abstractions from empirical
traces. We therefore do not claim trace-to-DTMC abstraction as the
main theoretical novelty. Our focus is its reliability-engineering use
for LLM agents: \textsc{TraceToChain} fits an absorbing DTMC with a
composite AIC$\,\wedge\,$KS goodness-of-fit certificate and uncertainty
quantification. Simulation study SS9 (\S\ref{sec:ss9}) evaluates
controlled held-out recovery on a held-out half of controlled
MAST-style trace corpora using KS distance and
$L_\infty^{\mathrm{RDC}}$ error. The first-passage view explains why
pass$@k$, pass$^k$, and the RDC are restrictions of a single
distribution (Proposition~\ref{thm:unify}), while the Goel--Okumoto
NHPP appears as a rare-failure scaling limit of the fitted chain
(Proposition~\ref{thm:nhpp}). This positions the work as reliability
evaluation for LLM agents, not new Markov-chain theory.

\noindent\textbf{Positioning summary.}
Table~\ref{tab:method-comparison} summarizes how these traditions
align with \textsc{TraceToChain}. The comparison is taxonomic rather
than competitive: classical reliability models provide semantics,
agent benchmarks provide traces, and probabilistic-verification tools
provide model-based queries.

\textsc{TraceToChain}'s distinction is the audited combination:
traces define the estimate, diagnostics test the abstraction,
uncertainty bounds the reported quantities, and first-passage
semantics connect horizon reliability, perturbations, and common
LLM-agent metrics. The result is a reliability argument from trace data
rather than another scalar score.

%% ================================================================

\section{Formalism}
\label{sec:formal}
\label{sec:preliminaries}

Deployment questions require probabilities that can be interpreted from
traces. The operational question is simple: after an agent has spent
several steps reasoning, calling tools, and observing outputs, what is
the probability that it reaches success before failure? We represent
that ``think-act-observe'' loop with a fitted absorbing discrete-time
Markov chain (DTMC), following classical absorbing-chain and
software-reliability models~\cite{kemenysnell,cheung1980}. Intermediate
agent behavior becomes transient execution states, and the two terminal
outcomes are task success~$\oplus$ and fatal failure~$\ominus$.
Reliability is then a first-passage question: how soon a run reaches
the success absorber, and whether failure arrives first.

This view also gives benchmark metrics a common semantics.
pass$@k$~\cite{chen2021codex}, pass$^k$, and the reliability decay curve
(RDC)~\cite{rdc2025} are all restrictions of one success
first-passage distribution. Figure~\ref{fig:concept_translation}
shows the trace-to-state interpretation used by the formal model.
We write $\ST$ for the transient state set: the discrete labels used
for intermediate reasoning, tool-use, and observation behavior before
a run reaches one of the two absorbing outcomes.

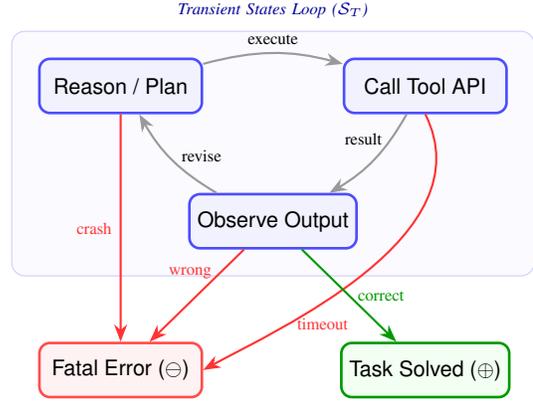
\begin{figure}[!t]
  \centering
  \begin{tikzpicture}[
      >=Stealth,
      scale=0.9, every node/.style={transform shape},
      concept/.style={rectangle, rounded corners, draw=blue!70, fill=blue!5, line width=1pt, align=center, font=\small\sffamily, minimum width=2.4cm, minimum height=0.8cm},
      success/.style={rectangle, rounded corners, draw=green!60!black, fill=green!5, line width=1pt, align=center, font=\small\sffamily, minimum width=2.4cm, minimum height=0.8cm},
      fail/.style={rectangle, rounded corners, draw=red!70, fill=red!5, line width=1pt, align=center, font=\small\sffamily, minimum width=2.4cm, minimum height=0.8cm},
      trans/.style={->, draw=gray!80, line width=0.8pt, font=\scriptsize}
  ]
      % Transient nodes
      \node[concept] (plan) at (0, 0) {Reason / Plan};
      \node[concept] (tool) at (4.5, 0) {Call Tool API};
      \node[concept] (obs) at (2.25, -2) {Observe Output};
      
      % Absorbing nodes
      \node[fail] (fail1) at (0, -4.2) {Fatal Error ($\ominus$)};
      \node[success] (succ) at (4.5, -4.2) {Task Solved ($\oplus$)};
      
      % Transient edges
      \draw[trans] (plan) to[bend left=15] node[above] {execute} (tool);
      \draw[trans] (tool) to[bend left=15] node[left, near start] {result} (obs);
      \draw[trans] (obs) to[bend left=15] node[right, pos=0.55, xshift=2pt] {revise} (plan);
      
      % Absorbing edges
      \draw[trans, color=red!80] (plan) -- node[left] {crash} (fail1);
      \draw[trans, color=red!80] (tool.south) .. controls +(0.9,-1.8) and +(1.1,0.6) .. node[right, pos=0.72] {timeout} (fail1.east);
      \draw[trans, color=red!80] (obs) -- node[left, near start] {wrong} (fail1);
      \draw[trans, color=green!60!black] (obs) -- node[right] {correct} (succ);
      
      % Background box
      \begin{scope}[on background layer]
          \node[draw=blue!20, fill=blue!2, rounded corners=5pt, fit=(plan)(tool)(obs), inner sep=10pt] (tbox) {};
      \end{scope}
      \node[text=blue!60!black, font=\scriptsize\itshape, above=2pt of tbox] {Transient States Loop ($\ST$)};
      
  \end{tikzpicture}
  \caption{Conceptual translation from an agent's interactive loop to
           a DTMC. Reasoning phases become transient states, and each
           run eventually absorbs into success or failure.}
  \label{fig:concept_translation}
\end{figure}

The fitted absorbing DTMC is a finite state-space summary of trace
dynamics. We denote it by
$\hat M=(\ST,\hat Q,\hat R_\oplus,\hat R_\ominus,s_0)$, where
$\hat Q$ contains estimated transitions among transient states,
$\hat R_\oplus$ and $\hat R_\ominus$ contain estimated exits to
success and failure, and $s_0$ is the initial state. Later sections
estimate these quantities from traces and test whether the abstraction
is adequate before using it for reliability claims.

\begin{definition}[Agent Markov chain]\label{def:amc}
An \emph{agent Markov chain} (AMC) is a tuple
$\hat M=(\ST,\hat Q,\hat R_\oplus,\hat R_\ominus,s_0)$
with the following components. $\ST$ is the finite set of transient
\emph{execution states} obtained from trace featurization
(\S\ref{sec:state_construction}). The matrix
$\hat Q\in[0,1]^{|\ST|\times|\ST|}$ gives transitions among those
states. The vectors
$\hat R_\oplus,\hat R_\ominus\in[0,1]^{|\ST|}$ give absorption
probabilities into task-complete state $\oplus$ and fatal-failure
state $\ominus$. Mass is conserved:
$\hat Q\mathbf{1}+\hat R_\oplus+\hat R_\ominus=\mathbf{1}$. Finally,
$s_0\in\ST$ is the initial state, or more generally an initial
distribution~$\pi_0$. Figure~\ref{fig:markov_chain} visualizes a
small AMC.
\end{definition}

\begin{definition}[Reliability]\label{def:reliab}
The $d$-step reliability is
$\mathcal R(d;\hat M)=\Pr[\tau_\oplus\le d\mid X_0=s_0]$, where
$\tau_\oplus$ is the first time the chain reaches the success
absorber. The asymptotic reliability is
$R_\infty=\lim_{d\to\infty}\mathcal R(d;\hat M)$. The Kemeny--Snell
fundamental matrix~\cite{kemenysnell} is $N=(I-\hat Q)^{-1}$. It
counts expected visits to transient states before absorption and
yields the closed form in Prop.~\ref{thm:closed}.
\end{definition}

This first-passage view also connects the fitted chain to classical
reliability growth. A central software-reliability model is the
non-homogeneous Poisson process (NHPP), a counting-process model for
failure discovery over time, especially the Goel--Okumoto
model~\cite{goelokt}. Musa--Okumoto gives a related reliability-growth
formulation~\cite{musa1987}. Under rare-failure scaling, the fitted
agent chain recovers this NHPP form as a limit
(Proposition~\ref{thm:nhpp}), which explains when classical
reliability-growth estimators are meaningful for agent traces.

\noindent\textbf{Model assumptions.}
The model makes four assumptions explicit. Transience means eventual
absorption and gives the fundamental matrix. Invertibility is the
algebraic condition needed for the closed forms. The initial condition
fixes where traces start, and i.i.d.\ replications are required when
interpreting pass$@k$ and pass$^k$ across repeated trials. The
first-order Markov property is not taken on faith: the
order-selection and goodness-of-fit tests in \S\ref{sec:order_selection}
and \S\ref{sec:gof} report when the trace corpus does not support the
fitted absorbing DTMC.

\begin{assumption}[Transience]\label{A:trans}
$\rho(\hat Q)<1$.
\end{assumption}
\begin{assumption}[Invertibility]\label{A:inv}
$(I-\hat Q)$ is invertible.
\end{assumption}
\begin{assumption}[Initial condition]\label{A:init}
$s_0$ is deterministic (extensions replace $e_{s_0}^\top$ by
$\pi_0^\top$).
\end{assumption}
\begin{assumption}[I.i.d.\ replications]\label{A:iid}
For pass$^k$ and pass$@k$, trials are i.i.d.\ unless noted.
\end{assumption}

\begin{definition}[Local perturbation family]\label{def:perturb}
For sensitivity analysis, let
$Q(\varepsilon)=Q_0+\varepsilon\Delta$ for
$0\le\varepsilon\le\varepsilon_{\max}$, where rows remain
substochastic and $\rho(Q(\varepsilon))<1$. The success-exit vector
$R_{\oplus,0}$ is held fixed, and any removed transient mass is
assigned to the failure absorber so that
$Q(\varepsilon)\mathbf{1}+R_{\oplus,0}+R_{\ominus}(\varepsilon)
=\mathbf{1}$. Operationally, this represents a local design change,
such as adding a fallback tool, without re-running the full benchmark.
\end{definition}

\begin{figure}[!t]
  \centering
  \begin{tikzpicture}[
      >=Stealth,
      scale=0.85, every node/.style={transform shape},
      transient/.style={circle, draw=blue!80, fill=blue!10, line width=1pt, minimum size=0.9cm, font=\small\bfseries},
      absorbing/.style={circle, draw=#1, fill=#1, fill opacity=0.15, text opacity=1, line width=1.5pt, minimum size=1cm, font=\small\bfseries, double, double distance=1.5pt, double=white},
      trans/.style={->, draw=gray!80, line width=0.8pt, font=\scriptsize},
      tosucc/.style={->, draw=green!60!black, line width=1pt, font=\scriptsize},
      tofail/.style={->, draw=red!70, line width=1pt, font=\scriptsize}
  ]
      % States
      \node[transient] (s0) at (0, 2.2)  {$s_0$};
      \node[transient] (s1) at (0, 0)    {$s_1$};
      \node[transient] (s2) at (0,-2.2)  {$s_2$};
      \node[absorbing=green!60!black] (succ) at (4.5, 1.1)  {$\oplus$};
      \node[absorbing=red!70]   (fail) at (4.5,-1.1)  {$\ominus$};

      % Transitions (Q)
      \draw[trans] (s0) edge[loop left, looseness=4, out=160, in=200] node[left] {0.5} (s0);
      \draw[trans] (s0) -- node[right] {0.3} (s1);
      \draw[trans] (s1) edge[loop left, looseness=4, out=160, in=200] node[left] {0.4} (s1);
      \draw[trans] (s1) -- node[right] {0.3} (s2);
      \draw[trans] (s2) edge[loop left, looseness=4, out=160, in=200] node[left] {0.6} (s2);

      % Transitions (R)
      \draw[tosucc] (s0) to[bend left=10]  node[above, sloped] {0.1} (succ);
      \draw[tofail] (s0) to[bend right=10] node[below, sloped] {0.1} (fail);
      \draw[tosucc] (s1) to[bend left=5]   node[above, sloped] {0.2} (succ);
      \draw[tofail] (s1) to[bend right=5]  node[below, sloped] {0.1} (fail);
      \draw[tosucc] (s2) to[bend left=10]  node[above, sloped] {0.3} (succ);
      \draw[tofail] (s2) to[bend right=10] node[below, sloped] {0.1} (fail);

      % Labels
      \begin{scope}[on background layer]
          \node[draw=blue!20, fill=blue!5, rounded corners=5pt, fit=(s0)(s1)(s2), inner sep=10pt] (tbox) {};
      \end{scope}
      \node[text=blue!60!black, font=\scriptsize\itshape, above=1pt of tbox] {Transient states $\mathcal{S}_T$};
  \end{tikzpicture}
  \caption{Illustration of a fitted absorbing DTMC. Transitions among
           transient states (matrix $Q$) model agent logic, while
           success and failure exit vectors $R_\oplus$ and
           $R_\ominus$ lead to absorbing states $\oplus$ and
           $\ominus$.}
  \label{fig:markov_chain}
\end{figure}
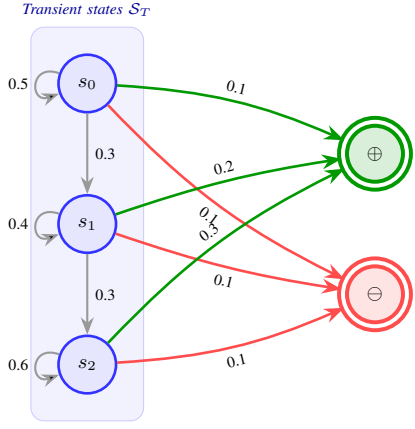

\noindent\textbf{Remark.}
The fitted absorbing DTMC is not intended to recover every hidden state
of the agent or its environment. It is a compact reliability model
estimated from observable traces, and its outputs are meaningful only
when the state construction and goodness-of-fit checks support the
first-order absorbing-DTMC approximation.

%% ================================================================

\section{State Construction from Traces}
\label{sec:state_construction}

The formalism is useful only if traces can produce the fitted chain
reproducibly and with checks that expose when the abstraction fails.
Trace tokens, tool calls, observations, and context become state labels,
transition counts, and exit probabilities that can be inspected. A
fitted Markov model is therefore an estimate plus evidence: the trace
representation, order test, first-passage fit, and uncertainty analysis
must all support its use.
Algorithm~\ref{alg:trace} implements the construction pipeline, and
Algorithm~\ref{alg:gof} gives the goodness-of-fit (GoF) protocol that
simulation study SS7 (\S\ref{sec:ss7}) checks on controlled corpora.

\subsection{Featurization and Discretization}
\label{sec:featurization}

The first step is to turn heterogeneous trace events into countable
states. In the running example, reasoning steps, tool calls, and
observations are featurized and then clustered into labels such as
``plan'' or ``tool result.'' Let
$\mathcal{T} = \{\tau^{(i)}\}_{i=1}^{n}$ be the trace corpus. A run is
$\tau^{(i)} = (\sigma_1^{(i)}, \sigma_2^{(i)}, \ldots,
\sigma_{L_i}^{(i)}, y^{(i)})$, where each $\sigma_t^{(i)}$ is one
observed step and $y^{(i)} \in \{\oplus, \ominus\}$ is the terminal
outcome. A featurizer $\phi : \sigma \mapsto \mathbb{R}^p$ maps each
step to a vector representation, either:
\begin{itemize}
  \item \emph{rule-based} (tool type, retry flag, error code, used in
        our simulation study SS7 and MAST experiments), or
  \item \emph{learned} (a pretrained sentence encoder applied to the
        agent's chain-of-thought and interchangeable in
        Algorithm~\ref{alg:trace}).
\end{itemize}
We cluster the feature vectors $\{\phi(\sigma_t^{(i)})\}$ using
agglomerative Ward linkage~\cite{ward1963}. The number of clusters
$m \in [k_{\min}, k_{\max}]$ is selected by silhouette
score~\cite{rousseeuw1987}. Each step is assigned to the selected
cluster, yielding
$s_1^{(i)}, \ldots, s_{L_i}^{(i)} \in \{1, \ldots, m\}$ followed by
the terminal outcome.

\subsection{Transition Estimation and Order Selection}
\label{sec:order_selection}

Once traces have state labels, estimating the fitted chain is a
counting problem with a direct operational interpretation. A row of
$\hat Q$ answers: after the agent is in this state, where does it go
next if it has not terminated? The corresponding entries of
$\hat R_\oplus$ and $\hat R_\ominus$ answer how often the same state
exits to success or fatal failure. We fit these quantities by
Laplace-smoothed maximum-likelihood estimation (MLE) with
$\alpha=1$, so
$\hat{Q}_{ij} = (c_{ij} + \alpha) / (c_i + \alpha(m+2))$
where $c_{ij}$ is the count of transitions from state $i$ to state
$j$, $c_{i,\oplus}$ and $c_{i,\ominus}$ count terminal exits from
state $i$, and
$c_i = \sum_j c_{ij} + c_{i,\oplus}+c_{i,\ominus}$. The same
denominator normalizes the success and failure exits. Smoothing
prevents sparse rows from assigning zero probability to unobserved but
plausible exits.

We also test whether first-order memory is sufficient. The comparison
uses a first-vs-second-order Markov Akaike information criterion (AIC)
test~\cite{akaike1974}:
$\Delta_{\rm AIC} = \mathrm{AIC}_2 - \mathrm{AIC}_1
                  = -2(\ell_2 - \ell_1) + 2(k_2 - k_1),$
where $\ell_o$ and $k_o$ are the log-likelihood and free-parameter
count under order $o$. The first-order model is preferred iff
$\Delta_{\rm AIC} > 0$. Simulation study SS7 (\S\ref{sec:ss7}) shows
that this test has high power against second-order ground truth.

\paragraph{Time-homogeneity as a testable assumption.}
When we pool transition counts, we assume a \emph{time-homogeneous}
kernel $\Pr[s_{t+1}=j \mid s_t=i]$ that does not vary with $t$. The
composite AIC$\,\wedge\,$KS protocol (Algorithm~\ref{alg:gof}) tests
whether that simplification is defensible. AIC checks whether recent
history still matters after conditioning on the current state, while
the Kolmogorov--Smirnov (KS) first-passage test can reveal drift in
exit hazards. If either test fails, the corpus should be segmented
before counts are pooled.

\paragraph{Variable horizons, early termination, and censoring.}
Runs can have different lengths and may end in success ($\oplus$),
failure ($\ominus$), or right-censoring. Algorithm~\ref{alg:trace}
counts each observed transient transition once. Completed runs add one
absorber count, while censored runs contribute only the transitions
observed before censoring. This keeps $\hat Q$ tied to observed
movement among transient states, but it can make the exit estimates
$\hat R_\oplus$ and $\hat R_\ominus$ conservative when censoring is
substantial. In \S\ref{sec:mast}, we observe censoring below $10\%$,
so this correction is a small source of conservatism rather than a
dominant driver of the fitted exits.

\begin{algorithm}[t]
\small
\DontPrintSemicolon
\KwIn{Traces $\mathcal{T}$, featurizer $\phi$, cluster range
      $[k_{\min}, k_{\max}]$, Laplace $\alpha$.}
\KwOut{$\hat{Q}, \hat{R}_\oplus, \hat{R}_\ominus$, labels $\pi$,
       AIC order verdict.}
\tcc{1. Featurize and vector-stack}
$X \leftarrow [\phi(\sigma_t^{(i)})]_{i,t}$\;
\tcc{2. Ward clustering with silhouette-selected $k$}
\For{$k = k_{\min}$ to $k_{\max}$}{
  $L_k \leftarrow$ AgglomerativeWard$(X, k)$\;
  $\mathrm{sil}_k \leftarrow $ silhouette$(X, L_k)$\;
}
$m^\star \leftarrow \arg\max_k \mathrm{sil}_k$;\quad
$\pi \leftarrow L_{m^\star}$\;
\tcc{3. Laplace-smoothed MLE}
\For{each transition $s_t^{(i)} \to s_{t+1}^{(i)}$}{
  $c_{s_t, s_{t+1}} \mathrel{+}= 1$\;
}
\For{each terminal exit $s_{L_i}^{(i)} \to y^{(i)}$}{
  $c_{s_{L_i}, y^{(i)}} \mathrel{+}= 1$\;
}
$\hat{Q}_{ij} \leftarrow (c_{ij} + \alpha) / (c_i + \alpha(m^\star + 2))$\;
$\hat{R}_{\oplus,i}, \hat{R}_{\ominus,i}$: same normalization against
terminal counts.\;
\tcc{4. Markov-order check}
Compute $\ell_1, \ell_2, k_1, k_2$;\quad
$\Delta_{\rm AIC} \leftarrow -2(\ell_2 - \ell_1) + 2(k_2 - k_1)$\;
\KwRet{$(\hat{Q}, \hat{R}_\oplus, \hat{R}_\ominus, \pi,
       \mathrm{1st\text{-}order if } \Delta_{\rm AIC} > 0)$}\;
\caption{\textsc{TraceToChain}: construct and audit an absorbing DTMC
         from agent execution traces.}
\label{alg:trace}
\end{algorithm}

\subsection{Goodness-of-Fit Protocol}
\label{sec:gof}

Fitting helps only if the abstraction matches the traces. A rejection
is an engineering result, not a nuisance: it says the current
state labels should not be used for reliability queries without
segmentation, new features, or a richer model. Given
$\hat{M} = (\hat{Q}, \hat{R}_\oplus, \hat{R}_\ominus)$, we use two
checks that address different failure modes:

\emph{Layer 1 (KS on first-passage).} We compute the analytic
success first-passage-time (FPT) cumulative distribution function
$F_{\tau_\oplus \mid \tau_\oplus < \tau_\ominus}^{\hat M}$,
conditional on eventual success, and compare it with the empirical
distribution of success FPTs among successful traces using a two-sample
KS test~\cite{smirnov1948}. The null is retained iff
$p_{\rm KS} > 0.05$.

\emph{Layer 2 (AIC).} The AIC layer asks whether the first-order
chain is adequate compared with a second-order alternative. We reject
the first-order model if $\Delta_{\rm AIC} < 0$.

\emph{Composite rule.} The chain is accepted iff
\emph{both} $p_{\rm KS} > 0.05$ \emph{and} $\Delta_{\rm AIC} > 0$.
This rule catches second-order dependence that the FPT-marginal KS test
can miss (see \S\ref{sec:ss7}). If either layer fails, later
reliability quantities should be treated as unsupported by the current
trace abstraction.

\begin{algorithm}[t]
\small
\DontPrintSemicolon
\KwIn{Traces $\mathcal{T}$, fitted chain $\hat{M}$, threshold
      $\alpha_{\rm KS}=0.05$.}
\KwOut{\textsc{Accept} / \textsc{Reject}, $(p_{\rm KS},\Delta_{\rm AIC})$.}
Compute conditional analytic CDF
$F_{\tau_\oplus \mid \tau_\oplus < \tau_\ominus}^{\hat M}$\;
Collect empirical FPTs from $\mathcal{T}_\oplus$ (success traces)\;
$p_{\rm KS} \leftarrow$ two-sample KS($\hat{F}$, $F^{\hat{M}}$)\;
$\Delta_{\rm AIC} \leftarrow$ (Algorithm~\ref{alg:trace}, step~4)\;
\eIf{$p_{\rm KS} > \alpha_{\rm KS}$ \textbf{and} $\Delta_{\rm AIC} > 0$}{
  \KwRet{\textsc{Accept}, $(p_{\rm KS}, \Delta_{\rm AIC})$}\;
}{
  \KwRet{\textsc{Reject}, $(p_{\rm KS}, \Delta_{\rm AIC})$}\;
}
\caption{Composite KS/AIC goodness-of-fit test for the agent-DTMC
         assumption.}
\label{alg:gof}
\end{algorithm}

\subsection{Deterministic Guarantees and Sources of Nondeterminism}
Given $\phi$, the cluster range, $\alpha$, and a silhouette
tie-breaking rule, Algorithm~\ref{alg:trace} is deterministic in the
trace corpus. Algorithm~\ref{alg:gof} uses no simulation. The
trace-generator seeds are included in the artifact submission.

\subsection{Uncertainty Quantification on \texorpdfstring{$\hat Q$, $\hat R_\oplus$, $\hat R_\ominus$}{Q, R+, R-}}
\label{sec:uq}

Algorithm~\ref{alg:trace} returns point estimates, but reliability
quantities inherit uncertainty from the fitted chain. A small trace
corpus may estimate a row of $\hat Q$ or an exit probability poorly,
and that uncertainty should be visible before a reviewer or operator
interprets $R(d)$, $R_\infty$, or a perturbation result. We attach
intervals to $\hat Q$, $\hat R_\oplus$, and $\hat R_\ominus$ and
propagate them to downstream quantities
(results~\ref{thm:closed}--\ref{thm:mixing}).

\paragraph{(i) Closed-form Dirichlet posterior.}
Dirichlet intervals quantify row-wise probability uncertainty. The
Laplace-smoothed MLE is the posterior mean under a symmetric
Dirichlet$(\alpha)$ prior over
$\{s_1,\ldots,s_m,\oplus,\ominus\}$. The row-$i$
posterior is
$\mathrm{Dir}(c_{i,\cdot} + \alpha \mathbf{1})$ and the marginal
for any single entry $X_{i,j}$ (including absorber columns) is
$\mathrm{Beta}\!\left(c_{i,j}+\alpha,\; c_i + \alpha(m+2) -
c_{i,j} - \alpha\right)$, where $c_i = \sum_{j'} c_{i,j'}$ includes
both transient and absorber outcomes. Equal-tailed
$(1-\alpha_{\rm CI})$ credible intervals (CIs) are the corresponding
Beta quantiles.

\paragraph{(ii) Non-parametric trace-level bootstrap.}
Bootstrap intervals quantify trace-sampling variability. We resample
traces with replacement $B$ times, assign each step to its nearest
target centroid to keep labels aligned, and recompute the MLE. The
implemented full re-clustering path (\texttt{fast=False}) gives
comparable but wider intervals.

\paragraph{Empirical ranges.}
Table~\ref{tab:mast-states} reports representative MAST-derived state
taxonomies. Cluster names come from dominant rule-based features
(\S\ref{sec:featurization}), and posterior intervals show which
success exits are well estimated or sparse. The differences are
descriptive rather than causal, but they expose state-level variation
hidden by aggregate pass/fail rates.

Table~\ref{tab:uq-summary} summarizes uncertainty on entries of
$\hat Q$. Posterior and bootstrap median CI widths are close across
frameworks, while larger maximum widths identify sparse or unstable
rows. All downstream $\mathcal{R}(d)$, mean time to absorption
($\mathrm{MTTA}$), and pass$^k$ quantities inherit these CIs by
propagation (Lipschitz-continuous in
$(\hat Q,\hat R_\oplus)$ by Theorem~\ref{thm:perturb}).

% Auto-generated by SS8_uncertainty_quantification.py
\begin{table}[t]
\centering
\caption{MAST-derived state taxonomy and 95\% posterior intervals for success-exit probabilities.}
\label{tab:mast-states}
\scriptsize
\setlength{\tabcolsep}{2pt}
\begin{tabular}{l l c c c}
\toprule
\textbf{Framework} & \textbf{Cluster} & \textbf{$\hat{R}_{\oplus,i}$} & \textbf{95\% post.\ CI} & \textbf{$\hat{R}_{\ominus,i}$} \\
\midrule
react~\cite{yao2023react} & tool\_call (100\%) & 0.064 & [0.039,0.094] & 0.067 \\
\tstriped
 & plan (100\%) & 0.076 & [0.048,0.109] & 0.072 \\
 & retry (100\%) & 0.034 & [0.004,0.092] & 0.085 \\
\tstriped
 & reflect (100\%) & 0.081 & [0.038,0.138] & 0.081 \\
 & error\_parse (100\%) & 0.082 & [0.028,0.162] & 0.082 \\
\tstriped
 & wait (100\%) & 0.057 & [0.007,0.153] & 0.057 \\
\midrule
reflexion~\cite{shinn2023reflexion} & plan (100\%) & 0.121 & [0.074,0.178] & 0.027 \\
\tstriped
 & error\_parse (100\%) & 0.083 & [0.041,0.138] & 0.033 \\
 & retry (100\%) & 0.054 & [0.020,0.103] & 0.045 \\
\tstriped
 & reflect (100\%) & 0.104 & [0.056,0.166] & 0.043 \\
 & tool\_call (100\%) & 0.097 & [0.066,0.133] & 0.064 \\
\tstriped
 & wait (100\%) & 0.128 & [0.049,0.236] & 0.064 \\
\midrule
toolformer~\cite{schick2023toolformer} & retry (100\%) & 0.105 & [0.057,0.164] & 0.048 \\
\tstriped
 & error\_parse (100\%) & 0.046 & [0.015,0.093] & 0.028 \\
 & tool\_call (100\%) & 0.078 & [0.053,0.108] & 0.062 \\
\tstriped
 & wait (100\%) & 0.073 & [0.021,0.154] & 0.036 \\
 & plan (100\%) & 0.092 & [0.050,0.145] & 0.050 \\
\tstriped
 & reflect (100\%) & 0.085 & [0.042,0.142] & 0.068 \\
\bottomrule
\end{tabular}
\end{table}

% Auto-generated by SS8_uncertainty_quantification.py
\begin{table}[t]
\centering
\caption{Per-framework 95\% uncertainty widths for entries of $\hat Q$.}
\label{tab:uq-summary}
\scriptsize
\setlength{\tabcolsep}{2pt}
\begin{tabular}{l c c c c c c}
\toprule
 & \multicolumn{3}{c}{\textbf{Posterior CI width}} & \multicolumn{3}{c}{\textbf{Bootstrap CI width}} \\
\cmidrule(lr){2-4}\cmidrule(lr){5-7}
\textbf{Framework} & \textbf{min} & \textbf{median} & \textbf{max} & \textbf{min} & \textbf{median} & \textbf{max} \\
\midrule
react~\cite{yao2023react} & 0.037 & 0.110 & 0.313 & 0.017 & 0.097 & 0.287 \\
\tstriped
reflexion~\cite{shinn2023reflexion} & 0.048 & 0.109 & 0.238 & 0.040 & 0.098 & 0.253 \\
cot\_agent~\cite{wang2024agentsurvey} & 0.029 & 0.111 & 0.337 & 0.008 & 0.084 & 0.288 \\
\tstriped
toolformer~\cite{schick2023toolformer} & 0.041 & 0.108 & 0.241 & 0.034 & 0.101 & 0.208 \\
babyagi~\cite{wang2024agentsurvey} & 0.040 & 0.113 & 0.286 & 0.035 & 0.100 & 0.291 \\
\tstriped
autogpt~\cite{wang2024agentsurvey} & 0.044 & 0.112 & 0.283 & 0.012 & 0.100 & 0.268 \\
agentbench~\cite{agentbench2024} & 0.034 & 0.116 & 0.265 & 0.032 & 0.103 & 0.250 \\
\bottomrule
\end{tabular}
\end{table}

The output is the audited model used for reliability analysis:
a fitted chain plus uncertainty and fit diagnostics. If AIC
or KS rejects the fit, the corpus should be segmented, re-featurized,
or treated as outside the absorbing-DTMC approximation. When both
checks pass, the fitted chain can be used for reliability queries with
the reported uncertainty intervals.

%% ================================================================

\section{Analytic Tools}
\label{sec:core}

Once a trace corpus has produced an accepted fitted chain, the next
question is what the chain can answer without another full benchmark
run. Three absorbing-chain tools address the deployment questions from
the introduction: reliability at a horizon, sensitivity to a local
change, and compatibility among common benchmark metrics.

\subsection{Closed-Form Reliability}

The first deployment question is reliability under a step budget:
starting at $s_0$, what is the chance of reaching $\oplus$ by step
$d$? In the running loop, this asks whether the agent reaches a
correct final answer within $d$ reasoning/tool/observation steps.

\begin{proposition}[Closed-form reliability; Kemeny--Snell~\cite{kemenysnell}]
\label{thm:closed}
Under Assumptions~\ref{A:trans}--\ref{A:init}, for every
$d\in\mathbb{N}_{\ge 0}$,
\begin{equation}
R(d) = e_{s_0}^\top(I-Q^{d})\,N\,R_\oplus,
\qquad N=(I-Q)^{-1},
\label{eq:t1finite}
\end{equation}
and
\begin{equation}
R_\infty = \lim_{d\to\infty}R(d)
         = e_{s_0}^\top N R_\oplus .
\label{eq:t1inf}
\end{equation}
\end{proposition}

\noindent\emph{Proof idea.}
The finite-horizon probability sums all paths that spend
$t=0,\ldots,d-1$ steps in transient states and then exit to success:
$R(d)=e_{s_0}^{\top}\sum_{t=0}^{d-1}Q^tR_\oplus$. The geometric
matrix identity
$\sum_{t=0}^{d-1}Q^t=(I-Q^d)(I-Q)^{-1}$ gives
\eqref{eq:t1finite}, and transience gives $Q^d\to0$, yielding
\eqref{eq:t1inf}. $N_{ij}$ is the expected number of visits to state
$j$ before absorption. Thus the closed form separates two effects an
operator can inspect: how often the agent visits each state and how
likely success is after those visits.

\subsection{Perturbation Sensitivity}

The second question is counterfactual: if a local transition changes,
how much can end-to-end reliability move? For example, a new fallback
tool may change one row of $Q$ by redirecting failed tool calls back
to planning. The fitted chain gives a first-order sensitivity
calculation before a full benchmark rerun is available.

\begin{proposition}[Perturbation sensitivity; Stewart~\cite{stewart1998}]
\label{thm:perturb}
Let $Q(\varepsilon)=Q_0+\varepsilon\Delta$ be the perturbation family
from Definition~\ref{def:perturb}. Write
$N_0=(I-Q_0)^{-1}$ and let $R_{\oplus,0}$ be the fixed success-exit
vector. Assume
$\varepsilon_{\max}\|N_0\|_\infty\|\Delta\|_\infty<1$.
For every $0\le\varepsilon\le\varepsilon_{\max}$,
\begin{equation}
\bigl|R_\infty(\varepsilon)-R_\infty(0)\bigr|
\le
\varepsilon \|N_0\|_\infty^2\|\Delta\|_\infty
        \|R_{\oplus,0}\|_\infty
 + C\varepsilon^2,
\label{eq:t2}
\end{equation}
where
\[
C=
\frac{\|N_0\|_\infty^3\|\Delta\|_\infty^2
      \|R_{\oplus,0}\|_\infty}
     {1-\varepsilon_{\max}\|N_0\|_\infty\|\Delta\|_\infty}.
\]
Moreover,
\begin{equation}
R_\infty(\varepsilon)-R_\infty(0)
=\varepsilon e_{s_0}^{\top}N_0\Delta N_0R_{\oplus,0}
 + O(\varepsilon^2).
\label{eq:t2lin}
\end{equation}
\end{proposition}

\noindent\emph{Proof idea.}
Expand $(I-Q(\varepsilon))^{-1}$ with the Neumann series around
$Q_0$. The smallness condition makes the remainder bounded by the
displayed constant. The first-order term is $N_0\Delta N_0$. The left
factor counts how often execution reaches the changed row, and the
right factor measures success value after the perturbed transition.
If $\|N_0\|_\infty$ is large, small local changes can have large global
effects and should be checked empirically. The bound is best read as a
screening calculation: it helps identify local changes that merit a
new benchmark run. As stated, the bound applies to perturbations of
$Q$ with $R_\oplus$ fixed. Changes to success exits require the
corresponding linear term for the exit vector.

\subsection{Metric Unification}
\label{sec:unify}

The third question is semantic. pass$@k$, pass$^k$, and the RDC often
appear to be competing summaries, but under the fitted chain they are
different views of one first-passage distribution. The disagreement is
usually about which projection to report, not about the underlying
reliability event. We keep the benchmark name reliability decay curve
(RDC), but the quantity modeled here is the cumulative success
probability by horizon $d$.

\begin{theorem}[Metric unification]
\label{thm:unify}
Let $R_\infty=e_{s_0}^{\top}NR_\oplus$. Under the i.i.d.\
replication assumption A\ref{A:iid},
\begin{align}
\mathrm{pass}^{k} &= R_\infty^{k}, \label{eq:t3-passk}\\
\mathrm{pass}@k  &= 1-(1-R_\infty)^{k}, \label{eq:t3-passatk}\\
\mathrm{RDC}(d)  &= R(d)
 = e_{s_0}^{\top}(I-Q^{d})NR_\oplus. \label{eq:t3-rdc}
\end{align}
Thus the three metrics are restrictions or transformations of the
same first-passage distribution.
\end{theorem}

\noindent\emph{Proof idea.}
The pass$^k$ identity is the probability that all $k$ independent
trials succeed. pass$@k$ is one minus the probability that all $k$
trials fail, and the RDC is exactly the finite-horizon reliability
from Proposition~\ref{thm:closed}.

The i.i.d.\ assumption is useful, but repeated agent runs can share
latent conditions such as a common prompt template, cache state, or
task difficulty. The following result makes that caveat measurable: it
shows how shared variation changes the two repeated-trial metrics even
when the marginal reliability remains $R_\infty$.

\begin{theorem}[Correlated-trial inequalities]
\label{thm:corr}
Suppose the $k$ trials share a latent state $\xi$ with
$\Pr[X_i=1\mid\xi]=p(\xi)$ and
$\mathbb{E}_\xi[p(\xi)]=R_\infty$. Conditional on $\xi$, the trials
are independent. Then
\begin{align}
\mathrm{pass}^{k} &= \mathbb{E}_\xi[p(\xi)^k]
                  \ge R_\infty^k, \label{eq:t3p-passk}\\
\mathrm{pass}@k  &= 1-\mathbb{E}_\xi[(1-p(\xi))^k]
                  \le 1-(1-R_\infty)^k. \label{eq:t3p-passatk}
\end{align}
For $k\ge2$, equality holds iff $p(\xi)$ is almost surely constant.
\end{theorem}

The inequalities show how hidden heterogeneity raises all-success
estimates and lowers at-least-one-success estimates relative to the
i.i.d.\ formulas.

\begin{corollary}[Diagnostic]
\label{cor:diag}
The gap
$\widehat{\mathrm{pass}^{k}}-\widehat{R_\infty}^{\,k}$ is a
lightweight diagnostic for hidden cross-trial correlation. Large
positive gaps indicate that repeated trials are not behaving as
i.i.d.\ samples from one fixed chain.
\end{corollary}

Large gaps should trigger inspection before repeated samples are
treated as independent trials. For reliability evaluation, the point is
practical: an assumption behind pass$@k$ and pass$^k$ becomes a
diagnostic rather than a hidden convention.

\noindent\emph{Proof idea.}
Condition on the latent variable and apply Jensen's inequality to
the convex functions $p^k$ and $(1-p)^k$.

For an accepted fitted chain, these results give horizon, perturbation,
and metric-reconciliation answers.

\section{Extensions}
\label{sec:ext}

The core tools answer immediate deployment questions. Three extensions
help interpret the answers: why an RDC may saturate, when a
classical reliability-growth curve emerges, and how many steps are
needed before $R(d)$ is close to $R_\infty$.

\subsection{RDC Shape}

An RDC is more useful when its shape has an explanation. In agent
traces, a concave curve suggests that early steps resolve the easier
cases and later steps add less new success probability. A convex
window suggests a barrier, such as needing a tool response before the
agent can make progress. The following condition captures the common
early-gains/saturation pattern.

\begin{assumption}[Stochastic monotonicity of $Q$]\label{A:monotone}
Fix an order $\preceq$ on the transient states. The matrix $Q$ is
\emph{stochastically monotone} if every non-increasing function
$f:\ST\to\mathbb{R}_{\ge0}$ remains non-increasing after one step,
i.e., $Qf$ is non-increasing whenever $f$ is. Equivalently, later
rows of $Q$ place stochastically more mass on later states in the
order~\cite{daley1968stochastic,keilsonKester1977}.
\end{assumption}

Reliability cannot decrease with more steps, and concavity follows when
success increments shrink over time.

\begin{proposition}[RDC shape; Keilson--Kester~\cite{keilsonKester1977},
Karlin--Rinott~\cite{karlin1980tp2}]
\label{thm:shape}
Let $v=e_{s_0}^{\top}$ and $w=R_\oplus$. For
$R(d)=v(I-Q^d)NR_\oplus$,
\begin{equation}
\Delta R(d)=R(d+1)-R(d)=vQ^dw\ge0,
\label{eq:t4-1stdiff}
\end{equation}
and, for $d\ge1$,
\begin{equation}
\Delta^2R(d)=R(d+1)-2R(d)+R(d-1)
           =-vQ^{d-1}(I-Q)w.
\label{eq:t4-2nddiff}
\end{equation}
If A\ref{A:monotone} holds, $w$ is non-increasing in the state order,
and $Qw\le w$ componentwise, then $R(d)$ is globally concave.
\end{proposition}

\noindent\emph{Proof idea.}
The first difference is the probability of exiting to success exactly
after $d$ transient steps. The second difference compares successive
exit probabilities. Because $vQ^{d-1}$ is nonnegative, the displayed
identity shows that $Qw\le w$ is enough to make the second difference
nonpositive. Stochastic monotonicity and the ordering of $w$ provide
the structural condition under which this shrinkage is expected.
A convex empirical window instead signals an early barrier or branch
where more steps initially increase eventual success.

\subsection{NHPP Rare-Failure Limit}

The absorbing-chain view also explains when aggregate
reliability-growth curves appear from step-level behavior. This is not
a separate modeling story: it is a limiting case of the fitted-chain
view. The limit isolates the rare-failure regime where a one-state
chain reduces to Goel--Okumoto.

\begin{proposition}[NHPP rare-failure limit; Goel--Okumoto~\cite{goelokt}]
\label{thm:nhpp}
Consider a sequence of one-transient-state agent chains with per-step
success probability $\mu_n$, per-step failure probability
$\varepsilon_n$, and
$Q_n=1-\mu_n-\varepsilon_n$. Suppose
\[
\begin{gathered}
\mu_n\to0,\qquad \varepsilon_n\to0,\qquad
d_n\mu_n\to\Lambda\in(0,\infty),\\
\varepsilon_n/\mu_n\to0.
\end{gathered}
\]
For any fixed $c>0$ and $d=cd_n$, the cumulative first-passage CDF
converges to
\begin{equation}
R_n(d)\to 1-e^{-c\Lambda}.
\label{eq:t5-GO}
\end{equation}
This is the Goel--Okumoto NHPP mean-value form
$m(t)=a(1-e^{-bt})$ with normalized $a=1$.
\end{proposition}

\noindent\emph{Proof idea.}
For the one-state chain,
$R_n(d)=\frac{\mu_n}{\mu_n+\varepsilon_n}
\left(1-(1-\mu_n-\varepsilon_n)^d\right)$. Under rare-failure
scaling, the prefactor tends to one and the geometric term tends to
$e^{-c\Lambda}$. Outside this low-error regime, the exact closed form
from Proposition~\ref{thm:closed} should be used.

\subsection{Approach-Rate Bound}

Even when $R_\infty$ is known, a practitioner still needs a step
budget. The next result converts spectral decay of $Q$ into a
conservative horizon rule: how many more agent steps are needed before
the remaining reliability gap is at most $\delta$?

\begin{proposition}[Approach-rate bound; Horn--Johnson~\cite{hornjohnson2013},
Stewart~\cite{stewart1998}, Levin--Peres~\cite{levin2017markov}]
\label{thm:mixing}
Let $\rho=\rho(Q)<1$ and $N=(I-Q)^{-1}$. If
$Q=V\Lambda V^{-1}$ is diagonalizable, then
\begin{equation}
R_\infty-R(d)
\le
\kappa_\infty(V)\rho^d\|NR_\oplus\|_\infty,
\label{eq:t6-diag}
\end{equation}
where $\kappa_\infty(V)=\|V\|_\infty\|V^{-1}\|_\infty$. Therefore a
sufficient horizon for error at most $\delta$ is
\begin{equation}
d^\star_{\mathrm{diag}}(\delta)=
\max\!\left\{0,\left\lceil
\frac{\log(\kappa_\infty(V)\|NR_\oplus\|_\infty/\delta)}
     {-\log\rho}
\right\rceil\right\}.
\label{eq:t6-steps-diag}
\end{equation}
For non-diagonalizable $Q$, the same exponential term is multiplied
by a polynomial factor determined by the largest Jordan block.
\end{proposition}

\noindent\emph{Proof idea.}
The gap is
$R_\infty-R(d)=e_{s_0}^{\top}Q^dNR_\oplus$, so bounding the gap is
the same as bounding $\|Q^d\|$. Diagonalization gives
$\|Q^d\|_\infty\le\kappa_\infty(V)\rho^d$. Jordan blocks add a
polynomial multiplier. Large $\rho(Q)$ or ill-conditioned eigenvectors
then imply longer step budgets.

These checks connect finite-horizon reliability, classical
reliability growth, and conservative step budgets within one
fitted-chain view.

%% ================================================================

\section{Simulation Validation}
\label{sec:sim}

Before applying the fitted chain to benchmark-shaped data, we check the
mechanics in controlled settings. Simulation study SS1 compares the
closed-form reliability formula with Monte Carlo simulation.
Figure~\ref{fig:analytic_checks} tests the perturbation and
correlated-trial claims, Figure~\ref{fig:goel_okumoto_limit} tests the
Goel--Okumoto rare-failure limit, and simulation study SS7 checks that
the GoF safeguard accepts and rejects controlled corpora for the right
reasons. These simulations do not validate unseen operational traces by
themselves. They show that the computation, analytic claims, and
diagnostics behave as intended on controlled inputs.

\subsection{SS1: Closed-Form vs.\ Monte Carlo}

SS1 isolates the numerical calculation behind the reliability estimates.
Before fitting traces, choosing clusters, or applying the GoF protocol,
we check whether the closed-form absorbing-chain expression for eventual
success agrees with direct simulation.

We generate 500 random substochastic chains per size
$m \in \{5, 10, 50, 100, 500\}$ and compare $R_\infty$
(Proposition~\ref{thm:closed}) with Monte Carlo estimates from $10^5$
trajectories per chain. Monte Carlo is used as an independent simulation
baseline for the same fitted chain, not as a competing estimator for
agent reliability.

\begin{figure}[!t]
  \centering
  \includegraphics[width=\columnwidth]{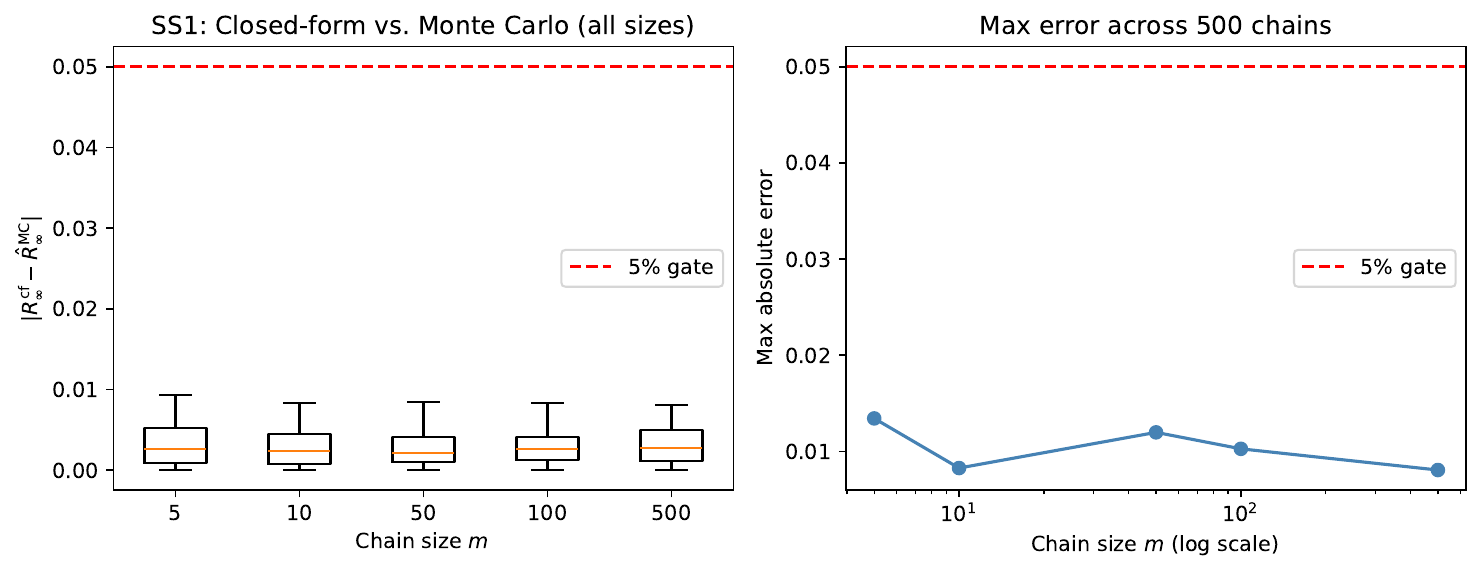}
  \caption{SS1 closed-form validation: absolute error
           $|R_\infty^{\rm cf} - \hat{R}_\infty^{\rm MC}|$ between
           analytic eventual-success reliability and Monte Carlo
           estimates across 500 chains per size. All errors are below
           the 5\% G2 threshold.}
  \label{fig:ss1}
\end{figure}

The agreement stays within threshold across all tested state-space
sizes. SS1 is a narrow but necessary check: it supports the
closed-form reliability computations used later in SS6 and SS9, while
leaving the Markov-fit question to the GoF protocol and held-out
validation.

Figure~\ref{fig:analytic_checks} checks how the fitted chain should be
used. Panel (a) supports the intended use of
Proposition~\ref{thm:perturb}: the perturbation bound screens local
changes by their direction and scale before a full rerun. Panels (b)
and (c) show that repeated attempts require an explicit dependence
assumption. Agreement, pass$^k$, and the independent-trial calculation
are different projections of first-passage behavior, not
interchangeable scalars.

\begin{figure}[!t]
  \centering
  \begin{minipage}[t]{0.32\columnwidth}
    \centering
    \includegraphics[trim=0bp 0bp 350bp 0bp,clip,width=\linewidth,height=0.075\textheight,keepaspectratio]{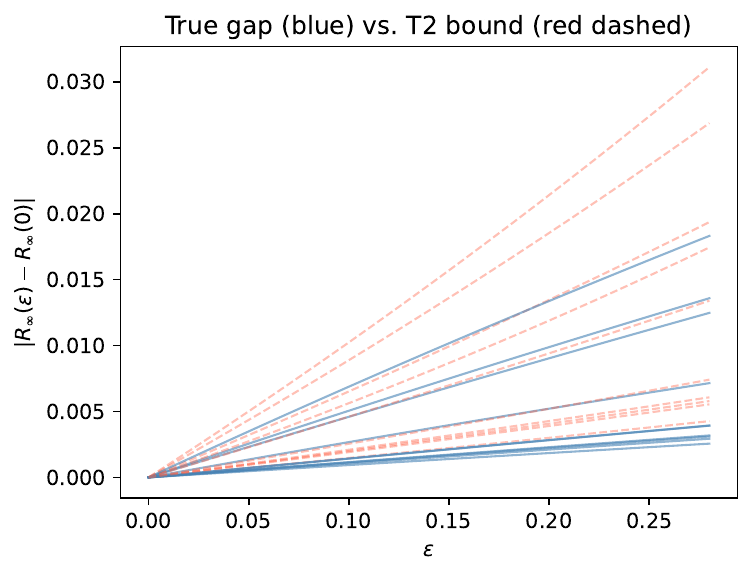}\\[-0.3ex]
    {\scriptsize (a) Perturbation}
  \end{minipage}\hfill
  \begin{minipage}[t]{0.32\columnwidth}
    \centering
    \includegraphics[trim=0bp 0bp 365bp 0bp,clip,width=\linewidth,height=0.075\textheight,keepaspectratio]{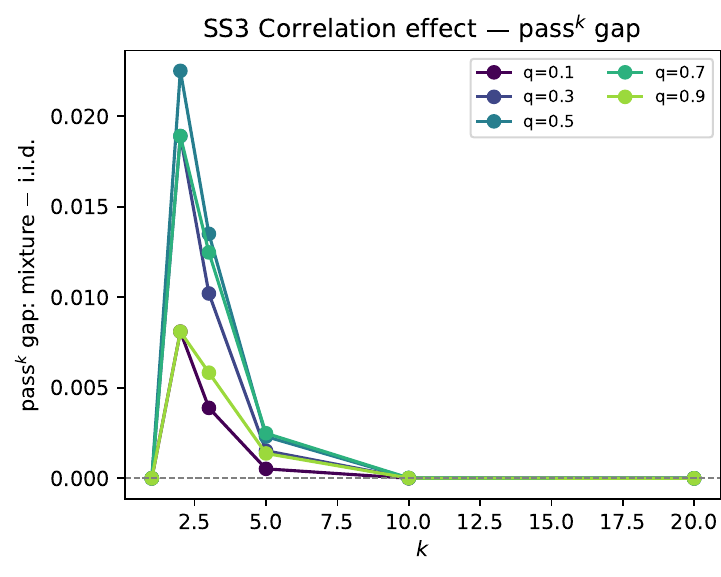}\\[-0.3ex]
    {\scriptsize (b) Agreement}
  \end{minipage}\hfill
  \begin{minipage}[t]{0.32\columnwidth}
    \centering
    \includegraphics[trim=365bp 0bp 0bp 0bp,clip,width=\linewidth,height=0.075\textheight,keepaspectratio]{fig_correlation_gap.pdf}\\[-0.3ex]
    {\scriptsize (c) pass$^k$}
  \end{minipage}
  \caption{Additional theorem checks retained in the artifact:
           (a) perturbation bound behavior, (b) correlated-trial
           agreement lift, and (c) pass$^k$ gap.}
  \label{fig:analytic_checks}
\end{figure}

Figure~\ref{fig:goel_okumoto_limit} keeps the Goel--Okumoto check
separate because the source plot contains six scaling regimes. The
insight is conceptual as much as numerical: when per-step absorption
probabilities become small and the horizon is rescaled, the
first-passage curve of the fitted chain approaches the classical NHPP
mean-value form. Thus Goel--Okumoto appears as a limiting case of the
agent-chain model, not as an unrelated historical baseline.

\begin{figure}[!t]
  \centering
  \includegraphics[width=\columnwidth,height=0.16\textheight,keepaspectratio]{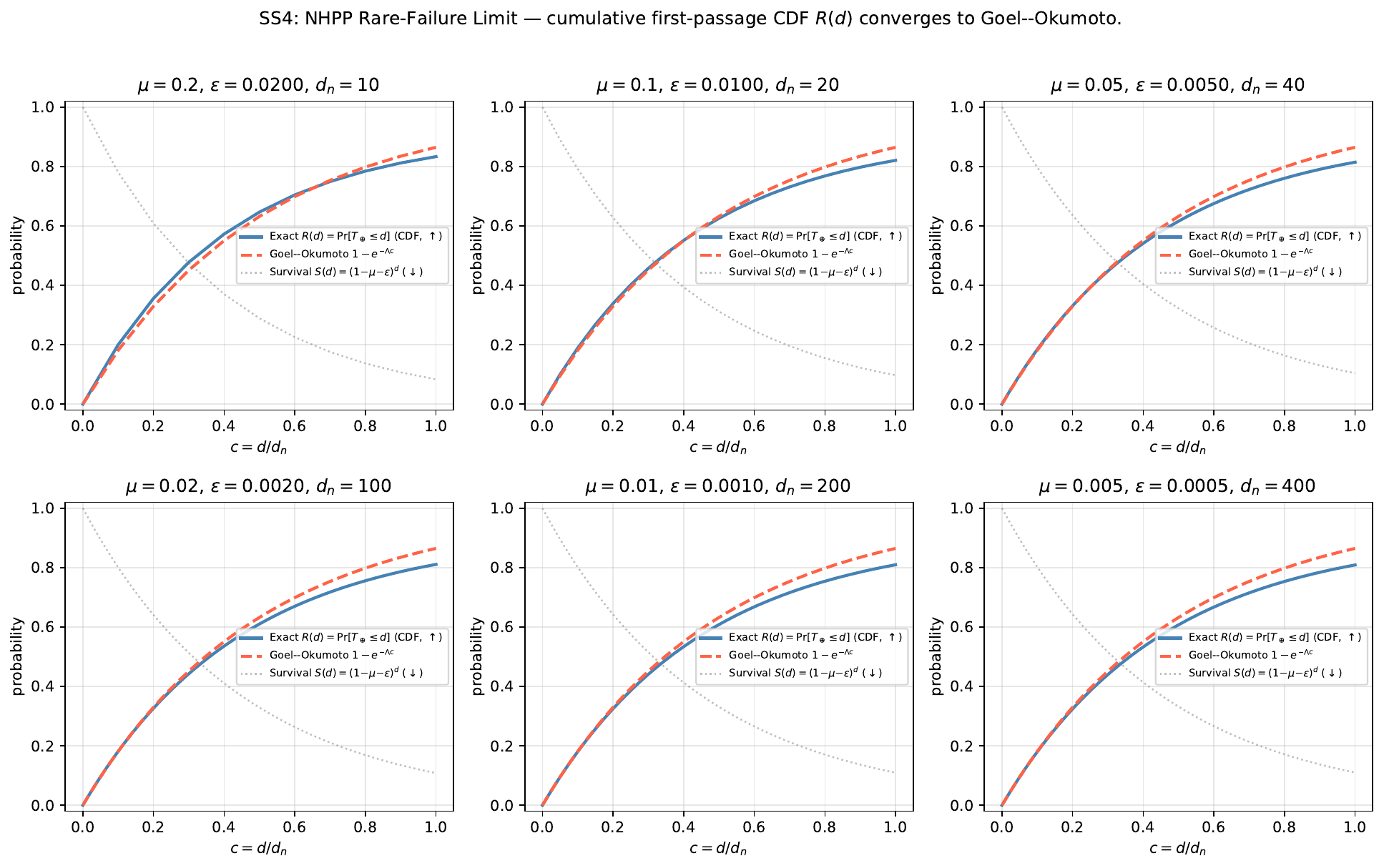}
  \caption{Goel--Okumoto rare-failure limit check: the cumulative
           first-passage distribution of the fitted-chain model
           approaches the NHPP mean-value form across six scaling
           regimes.}
  \label{fig:goel_okumoto_limit}
\end{figure}

\subsection{SS7: Goodness-of-Fit of the Agent-DTMC Assumption}
\label{sec:ss7}

SS7 tests the goodness-of-fit protocol of \S\ref{sec:gof}
(Algorithm~\ref{alg:gof}) on three controlled conditions. The goal is
not to make every trace Markov, but to separate corpora where the
fitted-chain abstraction is credible from those where it is not:

\textbf{(A) Markov ground truth.} On $30$ corpora of $300$ traces from
a known 1st-order absorbing chain ($m=5$), the composite test retains
the null in \textbf{$100\%$} of corpora at $\alpha=0.05$, confirming
Type-I control.

\textbf{(B) 2nd-order ground truth.} On $15$ corpora from a 2nd-order
chain with mixture weight $0.6$, the FPT-KS layer alone rejects $0\%$,
while the AIC layer rejects $100\%$ with
$\Delta_{\rm AIC}\in[-1540,-1310]$. The composite accept-if-both-pass
rule rejects $100\%$.

\textbf{(C) MAST-derived self-consistency.} We sample 500 synthetic
traces from each of the 7 MAST-derived chains used in \S\ref{sec:mast}.
The composite test accepts all 7 frameworks (KS $p$-values
$\in\{0.520, 0.583, 0.852, 0.944, 0.947, 0.957, 0.966\}$,
$\Delta_{\rm AIC}$ small-positive). Because these traces are sampled
from the fitted chain, this checks internal consistency, not whether
real agent traces are Markov. The held-out study then checks whether
the full trace-to-chain pipeline recovers first-passage behavior on
controlled MAST-style traces it did not fit.

\begin{figure}[!t]
  \centering
  \includegraphics[width=\columnwidth]{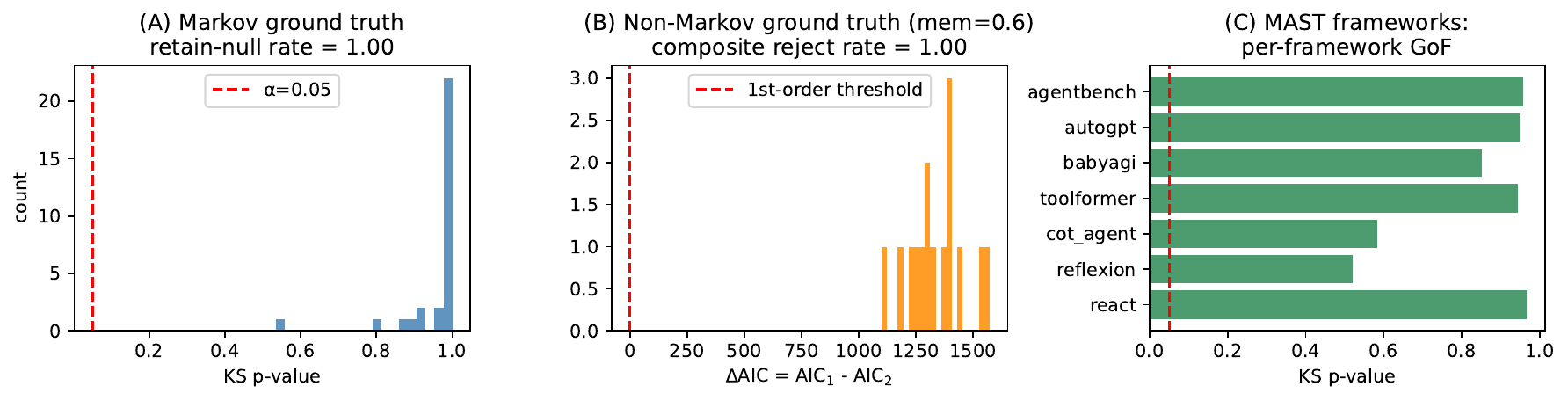}
  \\[-0.3ex]
  \begin{minipage}[t]{0.32\columnwidth}
    \centering
    {\scriptsize (a) Markov ground truth: KS $p$-values}
  \end{minipage}\hfill
  \begin{minipage}[t]{0.32\columnwidth}
    \centering
    {\scriptsize (b) 2nd-order ground truth: AIC rejection}
  \end{minipage}\hfill
  \begin{minipage}[t]{0.32\columnwidth}
    \centering
    {\scriptsize (c) MAST-derived: self-consistency KS}
  \end{minipage}
  \caption{SS7 tests the GoF safeguard across Markov ground truth,
           2nd-order ground truth, and MAST-derived self-consistency
           conditions.}
  \label{fig:gof}
\end{figure}

As a validation ladder, SS1 checks numerical computation,
Figures~\ref{fig:analytic_checks} and~\ref{fig:goel_okumoto_limit}
check analytic claims about perturbations, correlated trials, and
rare-failure limits, and Figure~\ref{fig:gof} checks the GoF
safeguard. The validation ladder separates numerical correctness,
analytic behavior, diagnostic behavior, and full held-out recovery of
MAST-style first-passage behavior.

%% ================================================================

\section{Empirical Case Studies}
\label{sec:case_studies}

The empirical evidence serves two different purposes. Simulation study
SS6 illustrates the reliability quantities on MAST-derived framework
summaries, so its role is descriptive. Simulation study SS9 uses a
strict held-out split on controlled MAST-style traces to test whether
\textsc{TraceToChain} recovers the RDC and success-time distribution on
unseen traces. This separation matters: SS6 shows what the reliability
vocabulary reports, while SS9 tests recovery under a known
data-generating process.

\subsection{MAST Benchmark Formulation}
\label{sec:mast}

SS6 uses transition matrices derived from public MAST summaries for
seven frameworks, applies the fitted-chain quantities, and ranks
frameworks by $R_\infty$. The exercise is illustrative: it does not
claim that raw MAST traces are Markov, but it shows what the
reliability vocabulary reports once a chain is available.
Figure~\ref{fig:rdc} shows each RDC.

\begin{figure}[!t]
  \centering
  \includegraphics[width=\columnwidth]{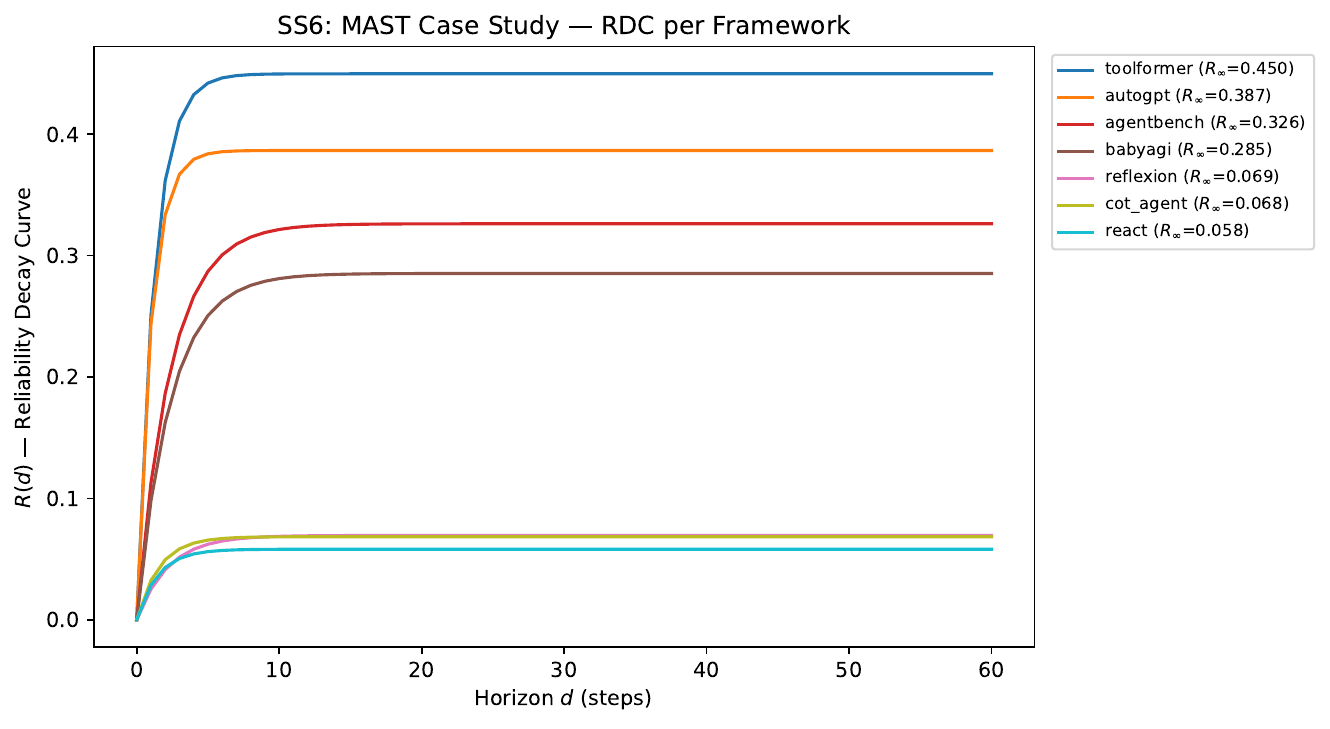}
  \caption{SS6 illustrates finite-horizon reliability on MAST-derived
           summaries: reliability decay curves for 7 frameworks ranked
           by $R_\infty$.}
  \label{fig:rdc}
\end{figure}

Table~\ref{tab:mast} summarizes the closed-form reliability features.
Toolformer has the highest asymptotic reliability in this illustrative
MAST-derived summary, while react and reflexion have lower overall
success.

\begin{table}[!t]
\centering
\caption{Closed-Form Reliability of MAST Frameworks.}
\label{tab:mast}
\begin{tabular}{lrrrr}
\toprule
\textbf{Framework} & \textbf{$R_\infty$} & \textbf{$\rho(Q)$} & \textbf{Horizon ($\delta=0.01$)} & \textbf{$m$} \\
\midrule
\thighlight
toolformer~\cite{schick2023toolformer}  & 0.4497 & 0.442 &  5 & 10 \\
\tstriped
autogpt~\cite{wang2024agentsurvey}     & 0.3866 & 0.370 &  4 &  5 \\
agentbench~\cite{agentbench2024}  & 0.3262 & 0.654 &  9 & 10 \\
\tstriped
babyagi~\cite{wang2024agentsurvey}     & 0.2853 & 0.656 &  8 & 12 \\
reflexion~\cite{shinn2023reflexion}   & 0.0695 & 0.636 &  5 & 12 \\
\tstriped
cot\_agent~\cite{wang2024agentsurvey}  & 0.0684 & 0.526 &  3 & 10 \\
react~\cite{yao2023react}       & 0.0581 & 0.509 &  3 &  9 \\
\bottomrule
\end{tabular}
\end{table}

\subsection{SS9: Held-Out Empirical Validation}
\label{sec:ss9}

Simulation study SS7(C) is an in-sample self-consistency check. SS9
instead evaluates controlled held-out recovery on MAST-style traces
with a strict fit/test protocol. The test set is not used to choose
features, clusters, transition estimates, or tuning, so the reported KS
and RDC errors measure held-out recovery rather than reuse of the
training corpus.

\paragraph{Protocol.}
For each of the 7 MAST-style frameworks we generate $n=400$
trajectories from a ground-truth absorbing chain whose transient
structure mimics the MAST taxonomy and emits noisy one-hot features
(Gaussian noise $\sigma=0.08$ and $\approx\!5\%$ censored traces). The
corpus is split \emph{before} featurization into
$n_{\mathrm{fit}}=200$ and $n_{\mathrm{test}}=200$. We run
\textsc{TraceToChain} on the fit half only.

Table~\ref{tab:heldout-ss9} reports recovered state count $m$, KS
agreement between the model success-time cumulative distribution
function (CDF)
$F_{\tau_\oplus \mid \tau_\oplus < \tau_\ominus}^{\hat M}$,
conditional on eventual success and sampled with $N{=}8{,}000$
trajectories, and the held-out success-time CDF
$\hat F_{\tau_\oplus}$ among successful traces. The table also reports
the sup-norm discrepancy
$L_\infty^{\mathrm{RDC}} =
\sup_{d\in[0,50]} |\mathcal R(d;\hat M) - \hat{\mathcal R}_{\mathrm{emp}}(d)|$.
These metrics are all first-passage checks: they ask whether the
fitted chain predicts when success occurs on traces it did not fit.

\paragraph{Empirical baseline.}
The held-out empirical RDC $\hat{\mathcal R}_{\mathrm{emp}}(d)$ is the
non-parametric baseline for observed horizons. The fitted chain is
useful only when it tracks that curve and passes GoF, and it adds
state-level interpretation, uncertainty propagation, perturbation
analysis, metric reconciliation, and horizon queries beyond the
observed grid.

% Auto-generated by SS9_heldout_empirical.py
\begin{table}[t]
\centering
\caption{Held-out validation on seven MAST-style trace corpora.}
\label{tab:heldout-ss9}
\small
\begin{tabular}{l c c c c c c}
\toprule
\textbf{Framework} & \textbf{$n_{\mathrm{fit}}$} & \textbf{$n_{\mathrm{test}}$} & \textbf{$m$} & \textbf{$D_{\mathrm{KS}}$} & \textbf{$p_{\mathrm{KS}}$} & \textbf{$L_\infty^{\mathrm{RDC}}$} \\
\midrule
react~\cite{yao2023react} & 200 & 200 & 5 & 0.017 & 1.000 & 0.048 \\
\tstriped
reflexion~\cite{shinn2023reflexion} & 200 & 200 & 5 & 0.031 & 0.992 & 0.052 \\
cot\_agent~\cite{wang2024agentsurvey} & 200 & 200 & 5 & 0.024 & 1.000 & 0.018 \\
\tstriped
toolformer~\cite{schick2023toolformer} & 200 & 200 & 5 & 0.024 & 1.000 & 0.052 \\
babyagi~\cite{wang2024agentsurvey} & 200 & 200 & 6 & 0.047 & 0.776 & 0.053 \\
\tstriped
autogpt~\cite{wang2024agentsurvey} & 200 & 200 & 6 & 0.032 & 0.987 & 0.039 \\
agentbench~\cite{agentbench2024} & 200 & 200 & 5 & 0.032 & 0.988 & 0.021 \\
\bottomrule
\end{tabular}
\end{table}

\begin{figure}[!t]
  \centering
  \includegraphics[width=\columnwidth]{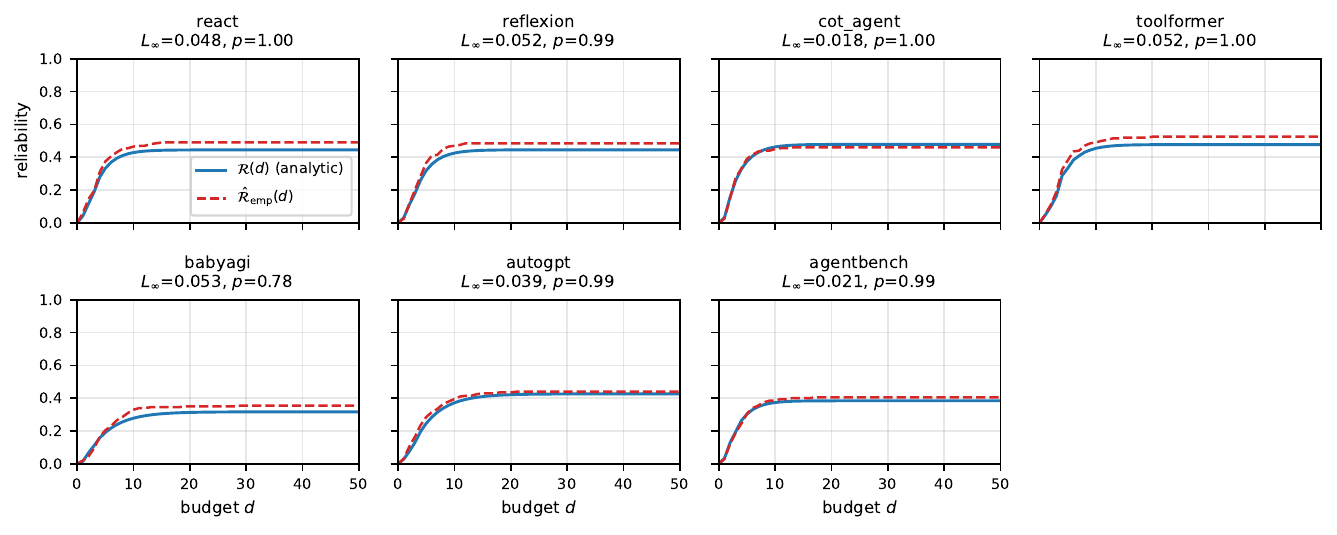}
  \caption{SS9 tests held-out recovery on controlled MAST-style traces:
           empirical RDC
           $\hat{\mathcal R}_{\mathrm{emp}}(d)$ (dashed) over analytic
           $\mathcal R(d)$ from the fitted chain (solid). Per-panel
           titles report $L_\infty^{\mathrm{RDC}}$ and $p_{\mathrm{KS}}$.}
  \label{fig:rdc_overlay}
\end{figure}

\paragraph{Findings.}
Across all 7 frameworks, the composite diagnostic passes on held-out
data at $\alpha=0.05$. KS distances are small
($D_{\mathrm{KS}}\in[0.017,0.047]$), the minimum KS $p$-value is
$0.776$, and $\max L_\infty^{\mathrm{RDC}}=0.053$ with median $0.048$.
BabyAGI has the largest RDC discrepancy yet still passes KS. This
supports recovery of the success first-passage distribution and RDC
for these controlled MAST-style traces.

\paragraph{Scope.}
SS9 uses synthetic traces from a genuine absorbing chain with noisy
emissions, so the held-out KS test probes controlled recovery by the
trace-to-chain pipeline:
featurization, clustering, Markov-order testing, and MLE. It does not
prove that operational agent traces are Markov. The MAST release lacks
the step-level features required by $\phi$, and SWE-bench /
$\tau$-bench trajectories require stable step-feature definitions
before the same held-out split can be applied.

\subsection{Cross-Benchmark Generalization}
\label{sec:cross_benchmark}

Agent evaluation spans diverse environments. To show that the
first-passage vocabulary is not tied to one taxonomy, we construct
synthetic archetypes for three benchmark shapes:
SWE-bench (software engineering)~\cite{swebench2024}, $\tau$-bench
(conversational tool-use)~\cite{taub2024}, and AgentBench
(multi-environment step-based reasoning)~\cite{agentbench2024}.
This is a portability exercise for the vocabulary, not empirical
validation on raw benchmark trajectories.

Table~\ref{tab:archetypes} formalizes the mapping from
benchmark-specific behaviors to the Markov state space $\ST$.

\begin{table}[!t]
\centering
\caption{Taxonomy Mapping for Cross-Benchmark Archetypes.}
\label{tab:archetypes}
\begin{tabular}{>{\raggedright\arraybackslash}p{0.22\columnwidth}
                >{\raggedright\arraybackslash}p{0.66\columnwidth}}
\toprule
\textbf{Benchmark} & \textbf{State Taxonomy ($\ST$)} \\
\midrule
SWE-bench~\cite{swebench2024} & \texttt{repo\_setup}, \texttt{issue\_read}, \texttt{search}, \texttt{edit\_file}, \texttt{test\_run} \\
\tstriped
$\tau$-bench~\cite{taub2024} & \texttt{user\_intent}, \texttt{api\_call}, \texttt{api\_resp}, \texttt{user\_clarify}, \texttt{confirm} \\
AgentBench~\cite{agentbench2024} & \texttt{think}, \texttt{act}, \texttt{observe} \\
\bottomrule
\end{tabular}
\end{table}

Using representative transition properties, we compute the same
analytic quantities. Figure~\ref{fig:cross_bench} shows characteristic
RDC behavior across the three archetypes.

\begin{figure}[!t]
  \centering
  \includegraphics[width=\columnwidth]{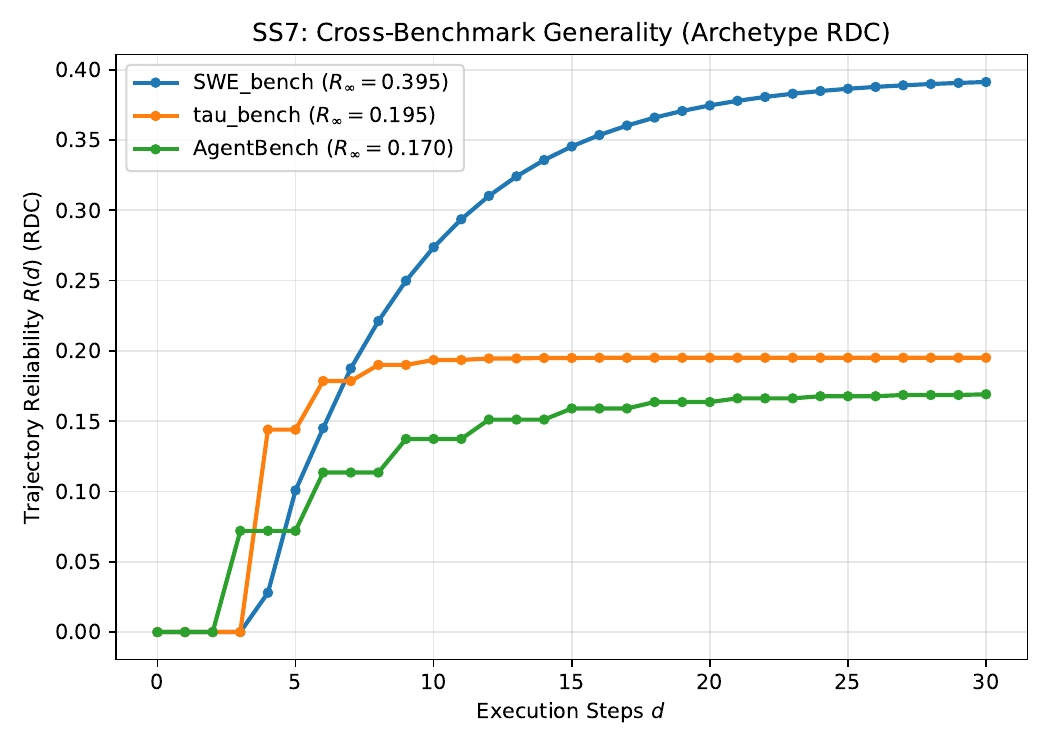}
  \caption{Cross-benchmark archetypes show vocabulary portability:
           RDCs for synthetic state spaces modeling SWE-bench,
           $\tau$-bench, and AgentBench.}
  \label{fig:cross_bench}
\end{figure}

These case studies support a narrower claim: when the diagnostics
accept the trace-to-chain abstraction, the fitted chain provides a
coherent reliability vocabulary for horizons, frameworks, and benchmark
shapes, with the limits set by the available trace representation and
the GoF diagnostics.

%% ================================================================

\section{Threats to Validity}
\label{sec:threats}

The main threat is overinterpreting the fitted abstraction. Diagnostics
and uncertainty reduce this risk, but the chain remains a reliability
model for specified trace features under specified diagnostics, not a
universal model of agent behavior. The estimates should therefore be
read conditionally: they depend on the state construction, the empirical
scope of the traces, the independence of repeated trials, the population
covered by the case studies, and the regime in which the NHPP limit is
used.

\paragraph{Construct validity (state construction).}
Mapping trajectories to a DTMC necessarily aggregates over memory, tool
state, and context, and no single mapping is canonical.
\S\ref{sec:state_construction} gives a reproducible construction whose
output is \emph{tested} by Algorithm~\ref{alg:gof}. Simulation study SS7
checks the protocol on known synthetic chains, and SS9 checks controlled
fit/test recovery. These studies establish statistical adequacy for the
chosen features, not uniqueness of $\phi$, independence of repeated
trials, or coverage of all raw benchmark distributions. If GoF rejects,
first-passage interpretation should stop until traces are segmented,
re-featurized, or modeled with a richer process. We therefore recommend
reporting $(\phi, m, p_{\mathrm{KS}}, \Delta_{\rm AIC})$ with each
estimate. Semi-Markov or hidden Markov models (HMMs) are natural next
steps when the DTMC abstraction is rejected.

\paragraph{Empirical scope of SS9.}
Simulation study SS9 uses synthetic traces from a true absorbing chain
with noisy one-hot emissions. The held-out KS test therefore evaluates
controlled recovery by the trace-to-chain pipeline
(featurizer + clustering + MLE), rather than the full modeling class on
unprocessed benchmark data. Applying the same split to
SWE-bench~\cite{swebench2024} and $\tau$-bench~\cite{taub2024} requires
step-level execution traces with stable feature definitions.

\paragraph{Independence assumption.}
The pass$^k$ and pass$@k$ identities require i.i.d.\ trials
(A\ref{A:iid}). This assumption can fail under temperature-0 sampling or
shared prefixes. Theorem~\ref{thm:corr} and Corollary~\ref{cor:diag}
quantify the resulting bias and provide a diagnostic.

\paragraph{External validity.}
The MAST case study covers 7 frameworks, so its conclusions may not
generalize to other agents, tasks, or trace distributions. For this
reason, SS6 is framed as illustrative rather than confirmatory.

\paragraph{NHPP-limit scope.}
The NHPP limit applies under rare-failure scaling
($\varepsilon\to 0$, $d\to\infty$). In high-failure regimes
($\varepsilon > 0.1$), the exact closed form from
Proposition~\ref{thm:closed} should be used instead.

Overall, the claims are conditional. The fitted chain is useful when the
trace representation is appropriate, diagnostics pass, and uncertainty
is acceptable for the deployment question. Otherwise, the fit should be
rejected or replaced by a richer model.

%% ================================================================

\section{Discussion}
\label{sec:disc}

For benchmark authors, SRE teams, and method developers, the fitted
chain provides a shared reliability object at the trace level. It
supports first-passage metrics with uncertainty, horizon and
perturbation queries, and diagnostics that can reject an inadequate
abstraction. Simulation study SS9 (Table~\ref{tab:heldout-ss9}) shows
that, once the trace-to-chain step is specified and checked, $\hat M$
tracks held-out first-passage behavior to within
$L_\infty^{\mathrm{RDC}}\le 0.053$ on controlled MAST-style trace
corpora for seven frameworks. This evidence does not replace
benchmarks or incident analysis. Rather, it shows how benchmark-style
trace corpora can support auditable reliability answers while also
identifying when the trace abstraction is too weak for those answers.
The practical requirement is to fit the chain, report uncertainty, run
GoF, and withhold first-passage conclusions when the checks fail.

Future work will extend the same audited first-passage view to
semi-Markov durations, HMM reliability under partial observability, and
online estimation of $Q$. We will also study integration with PRISM and
Storm and mixture-chain models for cross-trial
correlation~\cite{spaeh2024}.

%% ================================================================

\section{Conclusion}
\label{sec:conc}

LLM-agent teams already collect traces, but scalar metrics alone do not
answer deployment reliability questions. \textsc{TraceToChain} turns
those traces into an audited absorbing-chain model with
goodness-of-fit diagnostics and uncertainty. When the diagnostics
accept the trace abstraction, the fitted chain supports horizon
reliability, local perturbation analysis, and metric reconciliation.
The method combines a composite AIC$\,\wedge\,$KS certificate,
posterior and bootstrap intervals, and controlled MAST-style held-out
validation. On seven controlled MAST-style frameworks, analytic and
held-out empirical RDCs agree within
$\max L_\infty^{\mathrm{RDC}}=0.053$, and pass$@k$, pass$^k$, and the
RDC become projections of one distribution. The reliability claim is
therefore conditional but auditable: accepted trace abstractions give
deployment-facing first-passage answers, while rejected abstractions
should not be used for such conclusions. The main remaining step is to
apply the same split to SWE-bench and $\tau$-bench trajectories once
stable step-feature data are available.

%% ================================================================

\bibliographystyle{IEEEtran}
\bibliography{refs}

%% ----------------------------------------------------------------
%% [REDUCTION 2026-04-24] The Extended Proofs appendix is suppressed
%% from the 12-page body. All T1--T6 proofs and extended remarks are
%% retained in the reproducibility artifact under
%% \texttt{proofs/T1.tex}--\texttt{proofs/T6.tex}.
%% ----------------------------------------------------------------
\iffalse
\appendix
\section{Extended Proofs}
\label{app:proofs}

Full proofs of Propositions~\ref{thm:shape}--\ref{thm:mixing} are
provided in the companion artifact at
\texttt{proofs/T4.tex}--\texttt{proofs/T6.tex}. The main text
includes proof sketches only due to the 12-page limit.
\fi

\end{document}